\documentclass[onecolumn,draftcls]{IEEEtran}

\usepackage{latexsym}
\usepackage{amssymb}
\usepackage{bm}
\usepackage[dvips]{graphics}
\usepackage{graphicx}
\usepackage{psfrag}
\usepackage{amsfonts,amsmath,amssymb}
\usepackage{algorithm}
\usepackage{algorithmic}
\usepackage{dsfont,color,subfigure,setspace}

\usepackage{xspace}
\usepackage{bbm}
\newcommand{\Ac}{\mathcal{A}}

\newcommand{\Cc}{\mathcal{C}}
\newcommand{\Dc}{\mathcal{D}}
\newcommand{\Ec}{\mathcal{E}}

\newcommand{\Nc}{\mathcal{N}}

\newcommand{\Rc}{\mathcal{R}}

\newcommand{\Tc}{\mathcal{T}}
\newcommand{\Uc}{\mathcal{U}}
\newcommand{\Vc}{\mathcal{V}}

\newcommand{\Xc}{\mathcal{X}}
\newcommand{\Yc}{\mathcal{Y}}

\newcommand{\Xv}{{\bf X}}
\newcommand{\Yv}{{\bf Y}}

\newcommand{\Uv}{{\bf U}}

\newcommand{\xv}{{\bf x}}
\newcommand{\yv}{{\bf y}}

\newcommand{\jv}{{\bf j}}
\newcommand{\kv}{{\bf k}}

\newcommand{\Uh}{{\hat{U}}}

\newcommand{\uh}{{\hat{u}}}

\newcommand{\Ut}{{\tilde{U}}}

\newcommand{\Xt}{{\tilde{X}}}
\newcommand{\Yt}{{\tilde{Y}}}

\def\a{\alpha}
\def\b{\beta}

\def\d{\delta}
\def\e{\epsilon}

\DeclareMathOperator\E{E}
\let\P\relax
\DeclareMathOperator\P{P}
\newcommand\ie{i.e.,\xspace}
\def\textiid{i.i.d.\@\xspace}
\newcommand\iid{\ifmmode\text{ i.i.d. } \else \textiid \fi}

\newcommand{\ind}{\mathbbmss{1}}

\newtheorem{definition}{Definition}

\newtheorem{remark}{Remark}

\newtheorem{theorem}{Theorem}
\newtheorem{lemma}{Lemma}

\newtheorem{example}{Example}

\begin{document}
%\doublespacing

\title{Separation of source-network coding and channel coding in wireline networks}

\author{Shirin~Jalali,~\IEEEmembership{}
        and~Michelle Effros,~\IEEEmembership{Fellow,~IEEE}% <-this % stops a space
        \thanks{This paper was presented in part at IEEE International Symposium on Information Theory, Austin, Texas, 2010, and  Information Theory and Applications (ITA), San Diego, CA, 2011.}
\thanks{S. Jalali is with the Center for the Mathematics of Information, California Institute of Technology, Pasadena, CA 91125 USA (e-mail: shirin@caltech.edu),}
\thanks{M. Effros is with the Department of Electrical Engineering, California Institute of Technology, Pasadena, CA 91125 USA (e-mail:
effors@caltech.edu),}
}

\maketitle

\newcommand{\p}{\mathds{P}}
\newcommand{\Lc}{\mathcal{L}}
\newcommand{\mb}{\mathbf{m}}
\newcommand{\bb}{\mathbf{b}}
\newcommand{\Xb}{\mathbf{X}}
\newcommand{\Yb}{\mathbf{Y}}
\newcommand{\Ub}{\mathbf{U}}
\newcommand{\La}{\Lambda}
\newcommand{\su}{\underline{s}}
\newcommand{\xu}{\underline{x}}
\newcommand{\yu}{\underline{y}}
\newcommand{\Xu}{\underline{X}}
\newcommand{\Yu}{\underline{Y}}
\newcommand{\Uu}{\underline{U}}

\begin{abstract}
In this paper we prove the separation of source-network coding and channel coding in wireline networks.  For the purposes of this work, a wireline network is  any network of independent, memoryless, point-to-point, finite-alphabet channels used to transmit dependent sources either losslessly or subject to a distortion constraint. In deriving this result, we also prove that in a  general memoryless network with dependent sources, lossless and zero-distortion reconstruction are equivalent provided that  the conditional entropy of each source given the other  sources is non-zero. Furthermore, we extend the separation result to the case of continuous-alphabet, point-to-point channels such as additive white Gaussian noise (AWGN) channels.
\end{abstract}

%******************************************************************************
%******************************************************************************
\section{Introduction}

In his seminal work~\cite{Shannon:48},
Shannon separates the problem
of communicating a memoryless source
across a single noisy, memoryless channel
into separate lossless source coding
and channel coding problems.
The corresponding result for lossy coding in point-to-point channels is also proven in the same work. 
%The achievability part of the proof is   almost immediate since lossy coding in a point-to-point channel is equivalent to lossless coding of the codeword indices.  
For a single point-to-point channel,
separation holds under a wide variety of source and channel distributions
(see, for example,~\cite{VembuV:95} and the references therein).
Unfortunately, separation does not necessarily hold in network systems.
Even in very small networks like the multiple access channel~\cite{CoverE:80}, 
separation can fail when statistical dependencies
between the sources at different network locations
are useful for increasing the rate across the channel.
Since source codes tend to destroy such dependencies,
joint source-channel codes can achieve better performance
than separate source and channel codes in these scenarios.

This paper proves the separation between source-network coding
and channel coding in networks of independent noisy, discrete, memoryless channels (DMC); these networks are called {\em wireline networks} in this work.
Roughly, we show that the vector of achievable distortions
in delivering a family of dependent sources across such a network $\cal N$
equals the vector of achievable distortions
for delivering the same sources across a distinct network $\hat{\cal N}$.
Network $\hat{\cal N}$ is built
by replacing each channel $p(y|x)$ in $\cal N$
by a noiseless, point-to-point bit-pipe
of the corresponding capacity $C=\max_{p(x)}I(X;Y)$.
Thus a code that applies source-network coding
across links that are made almost lossless through the application
of independent channel coding across each link
asymptotically achieves the optimal performance
across the network as a whole.

Note that the operations of network source coding and network coding
are not separable, as shown in \cite{EffrosM:03}  and \cite{RamamoorthyJ:06}
for lossless source coding in non-multicast and multicast networks, respectively.
As a result, a joint network-source code is required, and only the channel code can be separated.
While the achievability of a separated strategy is straightforward,
the converse is more difficult since
preserving statistical dependence between
codewords transmitted across distinct edges of a network of noisy links
improves the end-to-end network performance
in some networks~\cite{KoetterE:09a,KoetterE_arxiv,KoetterE:11}.

The results derived here are consistent with those of \cite{Borade:02,SongY:06,KoetterE:11},
which prove the separation between network coding and channel coding
for multicast \cite{Borade:02,SongY:06} and general demands \cite{KoetterE:09a,KoetterE:11}, respectively,
under the assumption that messages transmitted to different subsets of users are independent.
The shift here is from independent sources  to dependent sources and from reliable information delivery  
 to both lossy and lossless data descriptions.

After hearing about our work, the author of~\cite{Yeung:10} pointed us to his unpublished work from the 90s, which proves the separation of lossy network source coding and channel coding in three specific network structures, namely, the Slepian-Wolf configuration, the multiple description configuration, and Yamamoto's cascade network. In these cases, \cite{Yeung:10} proves  separation without requiring the single-letter characterizations of the distortion regions. Our result  generalizes this result to any network configuration that consists of point-to-point noisy channels. The strategy underlying our proof follows that of \cite{KoetterE:09a,KoetterE:11}, but the details differ significantly, both due to the inclusion of dependent sources and lossy reconstruction and in the focus on discrete-alphabet channels. 

The organization of this paper is as follows.  Sections~\ref{sec:not} and~\ref{sec:setup} describe the notation and problem set-up, respectively.
Section~\ref{sec:stacked_net} describes a tool from \cite{KoetterE:11} called a stacked network that allows us to employ, in later arguments, typicality across copies of a network rather than typicality across time. Section~\ref{sec:results} proves the separation of lossy source-network coding and channel coding. Section \ref{sec:zero-dist} proves the equivalence of zero-distortion and lossless reconstruction in general memoryless channels. Section \ref{sec:awgn} shows that the separation of source-network coding and channel coding continues to hold  for well-behaved continuous channels such as AWGN channels under input power constraints.  Section \ref{sec:conclusion} concludes the paper.

The first part of the results presented in this paper, showing the separation of lossy source-network coding and channel coding in a wireline network  was first presented at ISIT 2010 \cite{JalaliE:10}. A similar result by other authors was presented at the same ISIT \cite{TianC:10}, where they prove that, in the same setup and under the finite source and channel alphabet assumption,  if each source is required only at one other node, or at multiple other nodes but at the same distortion level, then separation of source-network coding and channel coding is optimal. For the general case, under a restricted set of distortion measures,  they prove approximate optimality of separation strategy.

%******************************************************************************
%******************************************************************************
\section{Notation and definitions}\label{sec:not}

Finite sets are denoted by script letters such as $\Xc$ and $\Yc$. The size of a finite set $\Ac$ is denoted by $|\Ac|$. Random variables are denoted by upper case letters such as $X$ and $Y$. Bold face letters represent vectors. The alphabet of a random variable $X$ is denoted by $\Xc$. Random vectors are represented by upper case bold letters like  $\Xb$ and $\Yb$. The length of a vector is implied in the context. The $\ell^{\rm th}$ element of a vector $\Xv$ is denoted by $X_{\ell}$. A vector $\xv=(x_1,\ldots,x_n)$  or $\Xb=(X_1,\ldots,X_n)$ is sometimes represented as $x^n$ or $X^n$. For $1\leq i\leq j\leq n$, $x_{i}^j=(x_i,x_{i+1},\ldots,x_j)$.  For a set $\Ac\subseteq\{1,2,\ldots,n\}$, $\xv_{\Ac}=(x_i)_{i\in\Ac}$, where the elements are sorted in ascending order of their indices.

For two vectors $\xv,\yv\in\mathds{R}^r$, $\xv\leq \yv$ iff $x_i\leq y_i$ for all $1\leq i\leq r$. The $\ell_1$ distance between  two vectors $\xv$ and $\yv$ of the same length $r$ is denoted by $\|\xv-\yv\|_1 = \sum_{i=1}^r|x_{i}-y_{i}|$.
If $\xv$ and $\yv$ represent pmfs, i.e., $\sum_{i=1}^rx_{i}=\sum_{i=1}^ry_{i}=1$ and $x_{i}, y_{i}\geq 0$ for all $i\in\{1,\ldots,r\}$, then the total variation distance between $\xv$ and $\yv$ is defined as $ \|\xv-\yv\|_{\rm TV}=0.5\|\xv-\yv\|_1.$ 
\begin{definition}
The empirical distribution of a sequence $x^n\in\Xc^n$ is defined as
\begin{align*}
\pi(x|x^n)\triangleq {|\{i: x_i=x\}|\over n},
\end{align*}
for all $x\in\Xc$. Similarly,  the joint empirical distribution of  a sequence $(x^n,y^n)\in\Xc^n\times\Yc^n$ is defined as
\begin{align*}
\pi(x,y|x^n,y^n)\triangleq {|\{i: (x_i,y_i)=(x,y)\}|\over n},
\end{align*}
for all $(x,y)\in\Xc\times\Yc$.
\end{definition}
\begin{definition}
For a random variable $X\sim p(x)$ and a constant $\e>0$, the set $\Tc_{\e}^{(n)}(X)$ of $\e$-typical sequences\footnote{In this paper we only consider strong typicality, and use the definition introduced in \cite{OrlitskyR:01}.} of length $n$ is defined as 
\begin{align*}
\Tc_{\e}^{(n)}(X) \triangleq \{x^n: |\pi(x|x^n)-p(x)|\leq \e p(x) \;{\rm for }\;{\rm all}\; x\in\Xc \}.
\end{align*}
For $(X,Y)\sim p(x,y)$, the set $\Tc_{\e}^{(n)}(X,Y)$ of jointly $\e$-typical sequences is defined as 
\begin{align*}
\Tc_{\e}^{(n)}(X,Y) \triangleq \{(x^n,y^n):\;& |\pi(x,y|x^n,y^n)-p(x,y)|\leq \e p(x,y),\;{\rm for  }\;{\rm all}\; (x,y)\in\Xc\times\Yc \}.
\end{align*}
We shall use $\Tc_{\e}^{(n)}$ instead of  $\Tc_{\e}^{(n)}(X)$ or $\Tc_{\e}^{(n)}(X,Y)$ when the random variable(s) are clear from the context. 
\end{definition}
For $x^n\in\Tc_{\e}^{(n)}$, let 
\begin{align*}
\Tc_{\e}^{(n)}(Y|x^n)\triangleq\{y^n:(x^n,y^n)\in\Tc_{\e}^{(n)}\}.
\end{align*}

%******************************************************************************
%******************************************************************************

\section{The problem setup}\label{sec:setup}

Consider a multiterminal network $\Nc$ consisting of $m$ nodes  interconnected via a collection of point-to-point, independent DMCs. The network structure is represented by a directed graph $G$ with node set $\Vc$ and edge set $\Ec$. Each directed edge $e=(a,b)\in\Ec$ represents an independent  point-to-point DMC $(\Xc_e,p(y_e|x_e),\Yc_e)$ between nodes $a$ (input) and $b$ (output). For the channel represented by the edge $e$, the transition probabilities are $\{p(y_e|x_e)\}_{(x_e,y_e)\in\Xc_e\times\Yc_e}$. The channels are independent by assumption, together giving a multiterminal channel $(\prod_{e\in\Ec}\Xc_e,\prod_{e\in\Ec}p(y_e|x_e),\prod_{e\in\Ec}\Yc_e)$. The channel input at each node $a\in\Vc$ is $x^{(a)}=(x_{(a,v)}: (a,v)\in\Ec)$. The channel output at node $a$ is $y^{(a)}=(y_{(v,a)}:(v,a)\in\Ec)$.

 \begin{figure}
\begin{center}
\psfrag{p}{\scriptsize $p(y|x)$}
\psfrag{a}{\footnotesize $a$}
\psfrag{y1}{\scriptsize $Y_1^{t-1}$}
\psfrag{y2}{\scriptsize $Y_2^{t-1}$}
\psfrag{x1}{\scriptsize $X_{1,t}$}
\psfrag{x2}{\scriptsize $X_{2,t}$}
\psfrag{x3}{\scriptsize $X_{3,t}$}
\psfrag{Ua}{\scriptsize $U^{(a),L}$}
\mbox{\subfigure[Graph representation of a wireline network ]{\includegraphics[width=3.5cm]{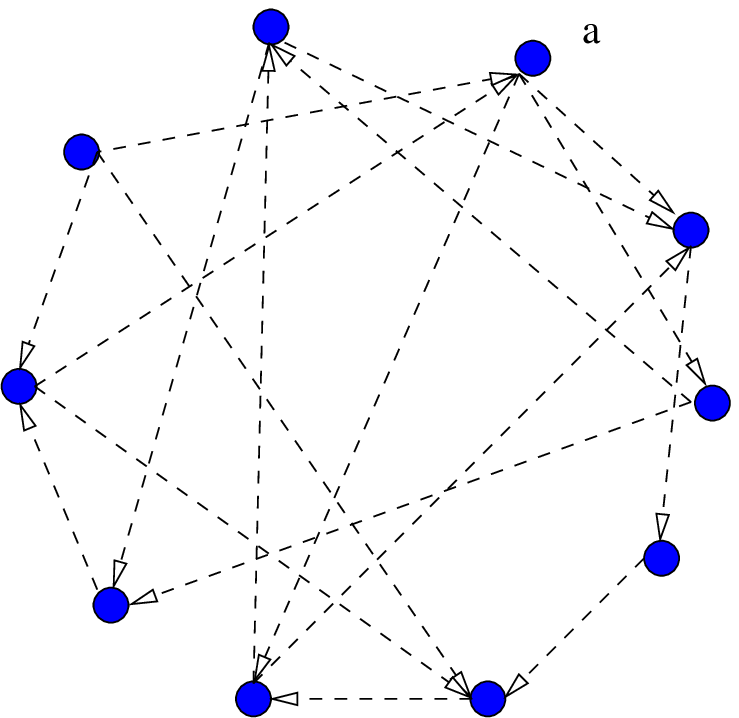}}\hspace{.5cm}
\subfigure[Each arrow represents a DMC.]{\includegraphics[width=5.5cm]{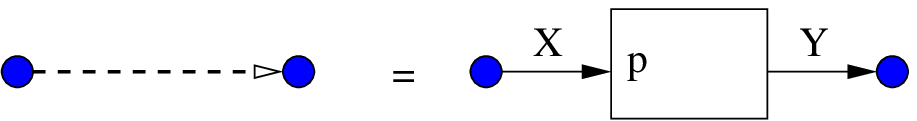}}\hspace{.5cm}
\subfigure[Coding operation at node $a$ at time $t$]{\includegraphics[width=4.8cm]{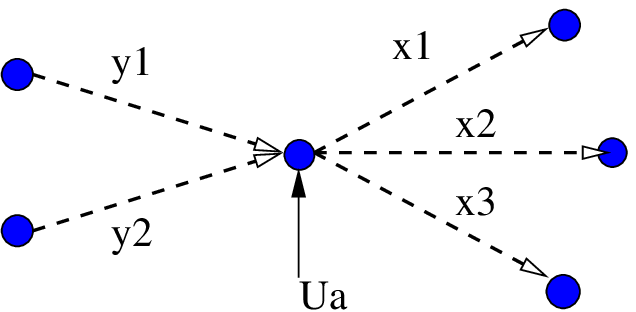}}}\caption{Noisy wired network model}\label{fig:wired-noisy}
\end{center}
\end{figure}

Each node $a$  observes some source process $\mathbf{U}^{(a)}=\{U^{(a)}_k\}_{k=1}^{\infty}$ and is interested in reconstructing the processes observed by  a subset of the other nodes. The alphabet  $\Uc^{(a)}$ of source $\mathbf{U}^{(a)}$  can be either scalar- or vector-valued. A vector-valued source  $\mathbf{U}^{(a)}$ denotes a collection of sources available at node $a$. In a block coding framework, source output symbols are divided into non-overlapping blocks of length $L$. Each block  is described separately. At the beginning of the $j^{\rm th}$ coding period, each node $a$  observes a length-$L$ block of the process $\mathbf{U}^{(a)}$, i.e., $U^{(a),j L}_{(j-1)L+1}=(U^{(a)}_{(j-1)L+1},\ldots,U^{(a)}_{j L})$. The blocks $\{U^{(a),j L}_{(j-1)L+1}\}_{a\in\Vc}$ observed at the nodes $a\in\Vc$ are described over $n$ uses of the network. The rate \[\kappa\triangleq\frac{L}{n}\] is a parameter of the code. At each time $t\in\{1,\ldots,n\}$, each node $a$  generates its next channel inputs as a function of  its source observation $U^{(a),j L}_{(j-1)L+1}$ and its observed channel outputs $Y^{(a),t-1}=(Y^{(a)}_1,\ldots,Y^{(a)}_{t-1})$ up to time $t-1$ using encoder 
\begin{align}
X_t^{(a)}:(\Yc^{(a)})^{t-1}\times\Uc^{(a),L}\to\Xc^{(a)}.\label{eq:Xt}
\end{align}
Note that each node might have more than one incoming channel and more than one outgoing channel. Thus,  $X_t^{(a)}$ and $Y_t^{(a)}$ are vectors with dimensions equal to the outdegree and indegree of node $a$, respectively. The reconstruction at node $b$ of the source vector observed at node $a$ is denoted by $\hat{U}^{(a\to b),L}$. This reconstruction is  determined using a decoder with inputs equal to the source and channel outputs observed at node $b$. Thus, $\hat{U}^{(a\to b),L}= \hat{U}^{(a\to b)}(Y^{(b),n},U^{(b),L})$, where
\begin{align}
\hat{U}^{(a\to b)}:\Yc^{(b),n}\times\Uc^{(a),L}\to\hat{\Uc}^{(a\to b),L}.\label{eq:Uhat_t}
\end{align}

The performance of   a  given code is the vector of expected average distortions between the sources $\{\Uv^{(a)}\}_{a\in\Vc}$ and reconstructions $\{{\bf \Uh}^{(a\to b)}\}_{a,b\in\Vc}$. For each $a,b\in\Vc$,
\begin{align*}
\E[d^{(a\to b)}_L(U^{(a),L},\hat{U}^{(a\to b),L})] \triangleq \E\left[\frac{1}{L}\sum\limits_{k=1}^Ld^{(a\to b)}(U^{(a)}_k,\hat{U}^{(a\to b)}_k)\right],
\end{align*}
where $d^{(a\to b)}:\Uc^{(a)}\times\hat{\Uc}^{(a\to b)}\rightarrow {\mathds{R}}^+$ is a per-letter distortion measure. As mentioned before $\Uc^{(a)}$ and $\hat{\Uc}^{(a\to b)}$ may be either scalar or vector-valued. This allows the case where node $a$ observes multiple sources and node $b$ is interested in reconstructing a subset of them.  Let 
\[
d_{\max}\triangleq\max\limits_{\substack{(a,b)\in \Vc^2\\ \a \in \Uc^{(a)},\b\in\hat{\Uc}^{(a\to b)}}} d^{(a\to b)}(\a,\b)<\infty.
\]
%If node $b$ is not interested in reconstructing node $a$, then  $d^{(a\to b)}\equiv 0$. 

The  $|\Vc|\times|\Vc|$ distortion matrix $D$ is said to be achievable at rate $\kappa$, if for any $\e>0$, and  any  $L$ large enough, there exists a blocklength-$(L,n)$ coding scheme such that
\[
 {L\over n} \geq \kappa-\e,
\] 
and 
\begin{align}
\E[d^{(a \to b)}_L(U^{(a),L},\hat{U}^{(a\to b),L})] \leq D(a,b) + \e,\label{eq:dist}
\end{align}
for every $(a,b)\in \Vc^2$. Let $\Dc(\kappa,\Nc)$ denote the set of achievable distortion matrices at  rate $\kappa$ in network $\Nc$.

\begin{remark}
While here we are assuming that all sources have a fixed rate $\kappa$, in general, the rate $\kappa$ can vary for different sources. Our results  continue to hold in that case as well. However, for notational simplicity, we assume that $\kappa$ is fixed among all sources. 
\end{remark}

Throughout the paper, for any network $\Nc$ of noisy point-to-point  channels described by  directed graph $G$,  let the network $\Nc_b$ denote a network of noiseless point-to-point channels described by the same directed graph $G$. Precisely,  network $\Nc_b$ replaces each noisy DMC $(\Xc_e,p(y_e|x_e),\Yc_e)$, $e\in\Ec$,  by a noiseless bit pipe of the same finite capacity $C_e=\max_{p(x_e)}I(X_e;Y_e)$. A bit pipe of capacity $C_e$ is an error-free, point-to-point communication channel that delivers, in $n$ channel uses, $\lfloor nC_e\rfloor$ bits  from the transmitter to the receiver, for any $n\geq 1$. The timing of the delivery of these bits has no impact on the set of achievable distortion matrices. This result is shown for the network  capacity problem in \cite{KoetterE:09a}; the same argument goes through immediately for the case of lossy reconstruction.

\begin{example}\label{exp:1}
Fig.~\ref{fig:exp1} demonstrates  a simple example of the kind of networks we study in this paper. The graph of the network, shown in  Fig.~\ref{fig2-1}, consists of two edges and three nodes. Each edge  models a  point-to-point  DMC. Fig.~\ref{fig2-2} shows a specific realization of such a network with sources available at nodes 1 and 2. It also shows how the encoding and  decoding operations are performed on network of Fig.~\ref{fig2-1}. The decoder reconstructs both sources $U_1$ and $U_2$.  For $i\in\{1,2\}$, let $C_i$ denote the noisy capacity of  channel $i$. Fig.~\ref{fig2-3} shows the equivalent noiseless model.  At coding rate $\kappa$,  choosing  $n=\lfloor \kappa^{-1} L\rfloor$,  $W_i\in\{1,\ldots,2^{\lfloor nC_i\rfloor}\}$. For this special example, \cite{Yeung:10} proves that at $\kappa=1$  separation is optimal and  the set of achievable distortions on both networks are equal. In this paper, we extend this result to general networks of point-to-point noisy channels, at arbitrary coding rate $\kappa$.
\end{example}

 \begin{figure}
\begin{center}
\psfrag{p1}{\footnotesize $p(y_1|x_1)$}
\psfrag{p2}{\footnotesize $p(y_2|x_2)$}
\psfrag{Y1}{\footnotesize $Y_1^n$}
\psfrag{Y2}{\footnotesize $Y_2^n$}
\psfrag{X1}{\footnotesize $X_1^n$}
\psfrag{X2}{\footnotesize $X_2^n$}
\psfrag{W1}{\footnotesize $W_1$}
\psfrag{W2}{\footnotesize $W_2$}
\psfrag{U1}{\footnotesize $U_1^L$}
\psfrag{U2}{\footnotesize $U_2^L$}
\psfrag{Uh}{\footnotesize $(\Uh_1^L,\Uh_2^L)$}
\psfrag{E1}{\footnotesize Encoder 1}
\psfrag{E2}{\footnotesize Encoder 2}
\psfrag{D}{\footnotesize Decoder}
\subfigure[Network model graph]{\includegraphics[width=2.5cm]{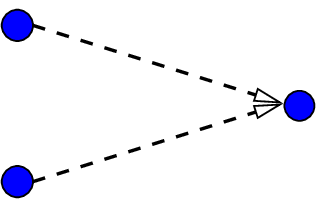}\label{fig2-1}}\\ \vspace{0.2cm}
\subfigure[Noisy  model]{\includegraphics[width=10cm]{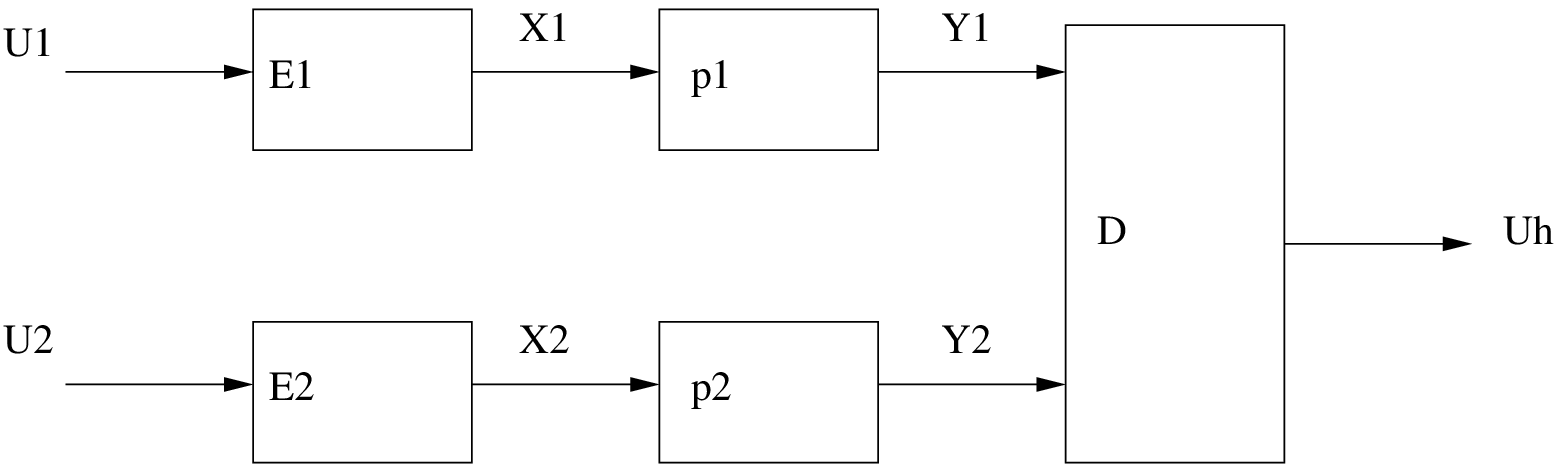}\label{fig2-2}}\\ \vspace{0.2cm}
\subfigure[Noiseless equivalent model]{\includegraphics[width=7.5cm]{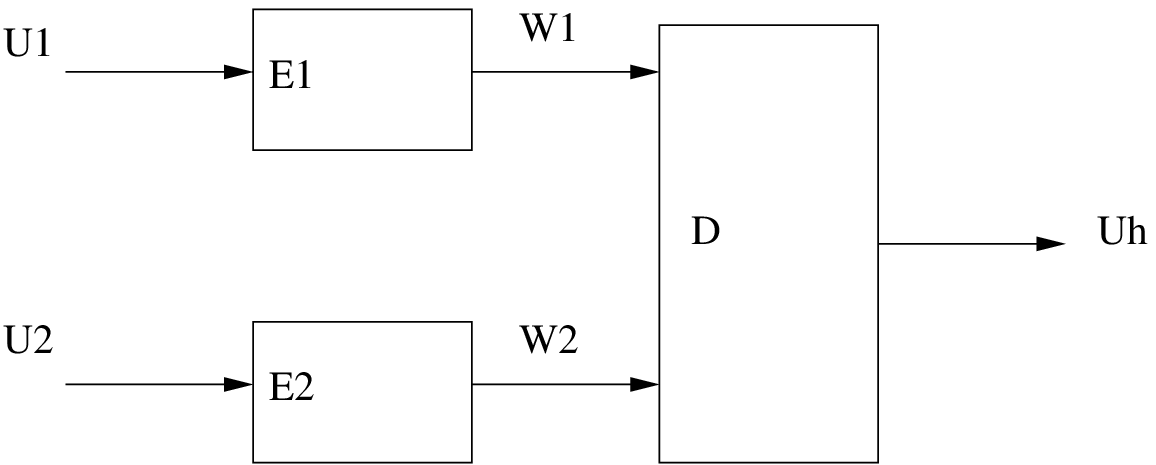}\label{fig2-3}}
\caption{Example \ref{exp:1}}\label{fig:exp1}
\end{center}
\end{figure}

%******************************************************************************
%******************************************************************************

\section{Stacked network} \label{sec:stacked_net}

The stacked network is a tool introduced in \cite{KoetterE:09a} for proving separation results. The key underlying observation is that by taking multiple copies of the same network and applying the same code to that network in each copy, we create i.i.d.~ copies of the input and output of a given channel at each time $t$. This allows us to later employ typicality arguments to our channel inputs and outputs across copies of the network and not across time. Applying typicality arguments across time is problematic since the inputs to the channel at different times $t$ need not be i.i.d.

For a given network $\Nc$, the corresponding $N$-fold stacked network $\underline{\Nc}$ is defined as $N$ copies of the original network \cite{KoetterE:09a}. That is, for each node $a\in\Vc$ and each edge $e\in\Ec$ in $\Nc$, there are $N$ copies of  node $a$ and $N$ copies of edge $e$ in  $\underline{\Nc}$. At each time instance, each node has access to the data available at all copies of node $a$, and each may use this extra information in generating the channel inputs for future time instances. Likewise, in decoding, all $N$ copies of a node can collaborate in reconstructing the source vectors. This is made more precise in the following two definitions. The encoder $\Xb_t^{(a)}$ for node $a$ at time $t$ in $N$-fold stacked network $\underline{\Nc}$ is a mapping 
\begin{align}
\Xb_t^{(a)}: \Yc^{(a),N(t-1)}\times\Uc^{(a),NL}\to \Xc^{(a),N},\label{eq:Xt_ul}
\end{align}
and the node-$b$ decoder $\hat{U}^{(a\to b),NL}$ for signal $U^{(a),NL}$ of node $a$ is a mapping
\begin{align}
\hat{U}^{(a\to b),NL}: {\Yc}^{(b),nN}\times\Uc^{(b),NL}\to\hat{\Uc}^{(a\to b),NL}.\label{eq:Uhat_t_ul}
\end{align}
These definitions  correspond to \eqref{eq:Xt} and \eqref{eq:Uhat_t} in network $\Nc$. In \eqref{eq:Xt_ul},  network input $\Xb_t^{(a)}$ is a vector with $N$ components denoted by  $\Xb_t^{(a)}=(\Xb_t^{(a)}(1),\ldots,\Xb_t^{(a)}(N))$.

In the $N$-fold stacked network, the distortion between the source originating  at node $a$ and its reconstruction at node $b$ is defined as
\begin{align*}
D_N(a,b)=\E \left[d^{(a\to b)}_{NL}(U^{(a),NL},\hat{U}^{(a\to b),NL})\right],
\end{align*}
for any $(a,b)\in\Vc\times\Vc$.

A distortion matrix $D$ is said to be achievable at rate $\kappa$ in the stacked version of network $\Nc$, if for any given $\e>0$, there exist $n$, $L$, and $N$ such that distortion $D$ and rate $\kappa$ are achievable in the $N$-fold stacked network; that is, $L/n\geq \kappa-\e$ and $D_N(a,b) \leq D(a,b)+\e$ for all $a,b\in\Vc$ on $N$-fold stacked network $\underline{\Nc}$.  Let $\Dc_{s}(\kappa,\underline{\Nc})$ denote the set of achievable distortion matrices at  rate $\kappa$ in the stacked network $\underline{\Nc}$. Note that the depth $N$ of the stacked network $\underline{\Nc}$ on which each distortion matrix  $D\in\Dc_s(\kappa,\Nc)$ is achievable may vary with $D$.

Note that the dimension of the distortion matrices in both single layer and multi-layer networks is $m\times m$. The following theorem establishes the relationship between the two sets.

\begin{theorem} \label{thm:stack}
At any rate $\kappa$, 
\begin{align}
\Dc(\kappa,\Nc)=\Dc_{s}(\kappa,\underline{\Nc}).
\end{align}
\end{theorem}
\begin{proof}[Proof of Theorem \ref{thm:stack}]

\begin{itemize}

\item[i.] $\Dc(\kappa,\Nc) \subseteq \Dc_s(\kappa,\underline{\Nc})$: This  is obvious, because the stacked network is a generalization of the original network. In fact,  choosing $N=1$, any distortion matrix $D$ that is achievable on $\Nc$ is also achievable on its stacked version too. 
%Consider any $D\in {\rm int}(\Dc(\kappa,\Nc))$. Then for any $\e>0$, there exists a code of rate $\kappa\leq L/n$ on $\Nc$ such that \eqref{eq:dist} is satisfied. For any $N$, a stacked network that uses this same coding strategy independently in each layer achieves expected distortion 
%    \begin{align*}
%    D_N(a,b)&=\E [ d^{(a\to b)}_{NL}(U^{(a),NL},\hat{U}^{(a\to b),NL})] \nonumber\\
%    &=\frac{1}{N} \sum\limits_{\ell=1}^N \E [d^{(a\to b)}_{L}(U^{(a), \ell L}_{(\ell-1)L+1},\hat{U}^{(a\to b),\ell L}_{(\ell-1)L+1})]\nonumber\\
%    &\leq\frac{1}{N}\sum\limits_{\ell=1}^N [D(a,b)+\e]\nonumber\\
%    &=D(a,b)+\e.
%    \end{align*}
%

\item[ii.] $\Dc_s(\kappa,\underline{\Nc}) \subseteq \Dc(\kappa,\Nc)$: The proof is very similar to the proof of the analogous part of Lemma 1 in \cite{KoetterE:11}, but,  for completeness,  we present the proof  in Appendix A.

\end{itemize}
\end{proof}

%******************************************************************************
%******************************************************************************

\section{Replacing a noisy channel with a bit pipe} \label{sec:results}

We assume that the sources are   independent and identically distributed (i.i.d.)  according to some distribution $p(u^{(1)},u^{(2)},\ldots,u^{(m)})$. That is, for any $k\geq1$, 
\begin{align*}
\P&\left((U^{(1),k},\ldots,U^{(m),k})=(u^{(1),k},\ldots,u^{(m),k})\right)\nonumber\\
&=\prod\limits_{i=1}^k p(u^{(1)}_i,\ldots,u^{(m)}_i).
\end{align*} 

For the given i.i.d. source assumption, Theorem \ref{thm:main} proves that the space of achievable distortions for networks $\Nc$ and $\Nc_b$ are identical. The proof follows the proof strategy of [6, Theorem 3], showing that any code for network $\underline{\Nc}_b$ can be applied across network $\Nc$ with the aid of a channel code and any code for $\Nc$ can be applied across network $\underline{\Nc}_b$ with the aid of an ``emulation code''. Just as a channel code enables us to emulate a noiseless bit pipe across a noisy channel, an emulation code enables us to emulate a noisy channel across a noiseless bit pipe. The result proves the optimality of separate source-network codes and channel codes on networks of point-to-point DMCs. Notice, however, that separate codes are here applied in the manner described in the proof of Theorem \ref{thm:stack} rather than the more conventional direct application across time.
  
\begin{theorem}\label{thm:main}
For a network $\Nc$ of independent point-to-point DMCs with memoryless sources,
\begin{align}
\Dc(\kappa,\Nc)=\Dc(\kappa,\Nc_b),
\end{align}
for any $\kappa>0$.
\end{theorem}

\begin{proof}[Proof of Theorem \ref{thm:main}]
By Theorem \ref{thm:stack}, the achievable region of a network $\Nc$ is equal to the achievable region of its stacked network $\underline{\Nc}$. Hence, $\Dc(\kappa,\Nc)=\Dc_s(\kappa,\underline{\Nc})$ and $\Dc(\kappa,\Nc_b)=\Dc_s(\kappa,\underline{\Nc}_b)$, and therefore, it suffices to prove that $\Dc_s(\kappa,\underline{\Nc})=\Dc_s(\kappa,\underline{\Nc_b})$.
\begin{itemize}
%Let $D\in{\rm int}(\Dc_s(\kappa,\underline{\Nc}_b))$. It suffices to show that $D\in\Dc_s(\kappa,\underline{\Nc})$. 
\item[i.] $\Dc_s(\kappa,\underline{\Nc}_b)\subseteq\Dc_s(\kappa,\underline{\Nc})$: Note that ${\Nc}$ and $\Nc_b$ are identical except that for each $e\in\Ec$, DMC $(\Xc_e,p(y_e|x_e),\Yc_e)$ in  $\Nc$  is replaced by a bit pipe of capacity $C=\max_{p(x_e)}I(X_e;Y_e)$ in $\Nc_b$. We next show that  any code for network $\underline{\Nc}_b$ can be  operated on $\underline{\Nc}$  with a similar expected distortion.  Fix any code of source blocklength $LN$, channel blocklength $nN$ and expected distortion matrix $D$ for $N$-fold stacked network $\underline{\Nc}_b$. Now consider a $pN$-fold stacked network  $\underline{\tilde{\Nc}}_b$. By partitioning the $pN$ layers into $p$ stacks, each consisting of $N$ layers, and then applying the  code independently to these stacks, we can construct a code for network $\underline{\tilde{\Nc}}_b$, which has the same  expected distortion matrix $D$. Consider a $pM$-fold stacked network $\underline{\Nc}$, with $M>N$.  Using the mentioned strategy to construct a code for the $pN$-fold stacked network from the code given for the $N$-fold network, at each time step $t$, each bit pipe $e$ in $\underline{\tilde{\Nc}}_b$ sends a message of at most $p\lfloor NC_e\rfloor$ bits across  the $pN$ copies of edge $e$ in $\underline{\tilde{\Nc}}_b$. To operate the same code on network $\underline{\Nc}$, we need to send the same information across the $pM$ copies of DMC $(\Xc_e,p(y_e|x_e),\Yc_e)$ in $\underline{\Nc}$. To achieve this goal, we use a channel code of blocklength $pM$ operating at rate $R_e<C_e$. By choosing $M= \lceil NC_e/R_e\rceil$, we guarantee that $pMR_e \geq pNC_e$. Hence  the $M$ copies of DMC $(\Xc_e,p(y_e|x_e),\Yc_e)$ in $\underline{\Nc}$ carry the same information as the $N$ copies of bit pipe $e$ in $\underline{\Nc}_b$. Since the capacity of  DMC $(\Xc_e,p(y_e|x_e),\Yc_e)$ equals $C_e$, $R_e$ can be made arbitrarily close to $C_e$. The rate of the code for $\underline{\Nc}$ is
\[
{ pNL \over pMn} ={ NL \over \left\lceil {NC_e / R_e}\right\rceil n}, 
\]
which can be made arbitrary close to $\kappa=L/n$.

Let $P^{({\rm M})}_{e}$ denote the maximal probability of error for  the channel code of blocklength $pM$ used over the $pM$ copies  of DMC $(\Xc_e,p(y_e|x_e),\Yc_e)$ in $\underline{\Nc}$. Let $P^{({\rm M})}_{\max}=\max_{e\in\Ec}P^{({\rm M})}_{e}$.  The code for each channel $e$ is used $n$ times $-$ once for each $t\in\{1,\ldots,n\}$.  Errors in the channel code for $e$ increase the distortion achieved by applying the code for $\underline{\tilde{\Nc}}_b$ across $\underline{\Nc}$. We can bound this  increase in  the expected average distortion using the union bound.   More precisely, let $\Rc$ denote the event that there is a decoding error in at least  one of the channels $e\in\Ec$ at some time step $t\in\{1,2,\ldots,n\}$. Since the sources and channel codes are independent, 
\begin{align*}
\E[d_{pNL}(U^{(a),pNL},\hat{U}^{(a\to b),pNL})] &= \E[d_{pNL}(U^{(a),pNL},\hat{U}^{(a\to b), pNL})|\Rc^c]\P(\Rc^c)\nonumber\\
&\; \;\;+  \E[d_{pNL}(U^{(a), pNL},\hat{U}^{(a\to b),pNL})|\Rc]\P(\Rc)\nonumber\\
& \leq D(a,b) + n|\Ec|P^{({\rm M})}_{\max}d_{\max},
\end{align*}
for each $(a,b)\in\Vc^2$. Therefore, for fixed $n$ and $N$, letting $p\to\infty$,  $|\E[d_{pNL}(U^{(a),pNL},\hat{U}^{(a\to b),pNL})] - D(a,b)|$ can be made arbitrarily small for each $(a,b)\in\Vc^2$.

\item[ii.] $\Dc_s(\kappa,\underline{\Nc})\subseteq\Dc_s(\kappa,\underline{\Nc}_b)$: Let $D\in \Dc(\kappa,\Nc)$. We prove that $D\in\Dc_s(\kappa,\underline{\Nc}_b)$. Consider a code defined on $\Nc$ with source blocklength $L$, channel blocklength $n$, and  an expected distortion matrix that is component-wise upper-bounded by   $D+\e\cdot\mathbf{1}$.   Applying this  code independently in each layer of $N$-fold stacked network $\underline{\Nc}$ gives a code for $\underline{\Nc}$ with $D_N(a,b)\leq D(a,b)+\e$, for all $(a,b)\in\Vc^2$. Throughout the rest of the proof, $n$ and $L$, corresponding to the source and channel blocklengths of the  mentioned code, are fixed. To simulate the performance of this code on the stacked version of $\Nc_b$, we let the number of layers $N$  go to infinity. As shown in \cite{KoetterE:11} any code for an $N$-fold stacked network can be unraveled across time to give a single-layer code with the same performance. The blocklength for that code goes to infinity as $N$ grows without bound.

We first show that for  identically distributed memoryless sources, the performance of the code given the realization of  $(\Xb_{e,1},\Yb_{e,1})$  depends only on the empirical distribution $\{\pi(x_e,y_e|\Xb_{e,1},\Yb_{e,1})\}_{(x_{e,1},y_{e,1})\in\Xc_e\times\Yc_e}$ of $(\Xb_{e,1},\Yb_{e,1})$. Here the subscript $1$ refers to time $t=1$. After establishing this, we use the result proved in \cite{CuffP:09} and show that at time $t=1$ we can emulate the behavior of the noisy link across a bit pipe of the same capacity. For the rest of the proof, let $\Ub=\{U_t\}$ denote an i.i.d. source observed at some node in $a\in\Vc$ and $\hat{\Ub}=\{\Uh_t\}$ denote its reconstruction at some other node $b\in\Vc\backslash\{a\}$.

In network $\Nc$, the expected distortion between source vector $U^L$ and its reconstruction $\Uh^L$ is
\begin{align}
&\E[d_L(U^{L},\hat{U}^{L})] =\nonumber\\ 
&\sum_{(x_{e,1},y_{e,1})\in\Xc_e\times\Yc_e}
 \E\left[d_L(U^{L},\hat{U}^{ L})\left|(X_{e,1},Y_{e,1})=(x_{e,1},y_{e,1})\right.\right]\nonumber\\
&\hspace{1.5cm}\times\P\left((X_{e,1},Y_{e,1})=(x_{e,1},y_{e,1})\right).\label{eq:single_layer_t1}
\end{align}

In the $N$-fold stacked network $\underline{\Nc}$,  the reconstruction of the corresponding $N$ independent copies of $U^L$ by reproduction $\Uh^{NL}$ satisfies 
\begin{align}
&\E \left[d_{NL}(U^{NL},\hat{U}^{NL})\right] \nonumber\\
&=\E\Big[{1\over N}\sum\limits_{\ell=1}^Nd_{L}\left(U^{\ell L}_{(\ell-1)L+1},\hat{U}^{\ell L}_{(\ell-1)L+1}\right) \times \nonumber\\
&\hspace{1cm} \sum_{(x_{e,1},y_{e,1})\in\Xc_e\times\Yc_e} \ind_{(\Xv_{e,1}(\ell),\Xv_{e,1}(\ell))=(x_{e,1},y_{e,1})} \Big]\nonumber\\
&={1\over N}\sum_{(x_{e,1},y_{e,1})\in\Xc_e\times\Yc_e}\sum\limits_{\ell=1}^N\nonumber\\
&\quad\quad\quad \E\left[d_{L}\left(U^{\ell L}_{(\ell-1)L+1},\hat{U}^{\ell L}_{(\ell-1)L+1}\right) \ind_{(\Xv_{e,1}(\ell),\Yv_{e,1}(\ell))=(x_{e,1},y_{e,1})}\right].\label{eq:i-1}
\end{align}
For any random variables $A$ and $B$, $\E[A\ind_{B=b}]=\sum_{a,b'}a\ind_{b'=b}p(a,b)=\sum_{a}ap(a,b)=p(b)\E[A|B=b]$.  Using this equality, and since the code used on $\underline{\Nc}$  applies the solution for $\Nc$ independently in each layer of  stacked network $\underline{\Nc}$, it follows that
\begin{align}
\E&\left[d_{L}\left(U^{\ell L}_{(\ell-1)L+1},\hat{U}^{\ell L}_{(\ell-1)L+1}\right) \ind_{(\Xv_{e,1}(\ell),\Yv_{e,1}(\ell))=(x_{e,1},y_{e,1})} \right] \nonumber\\
&=\E[d_{L}(U^{\ell L}_{(\ell-1)L+1},\hat{U}^{\ell L}_{(\ell-1)L+1})|(\Xv_{e,1}(\ell),\Yv_{e,1}(\ell))=(x_{e,1},y_{e,1})] \P((\Xv_{e,1}(\ell),\Yv_{e,1}(\ell))=(x_{e,1},y_{e,1}))
\nonumber\\
&=\E[d_{L}(U^L,\hat{U}^L)|(\Xv_{e,1}(\ell),\Yv_{e,1}(\ell))=(x_{e,1},y_{e,1})] P((\Xv_{e,1}(\ell),\Yv_{e,1}(\ell))=(x_{e,1},y_{e,1}))\label{eq:i-2}
\end{align}
where each conditional expectation of $d_L(U^L,\Uh^L)$ in \eqref{eq:i-2} equals the corresponding conditional expectation in \eqref{eq:single_layer_t1}.
Combining \eqref{eq:i-1} and \eqref{eq:i-2} yields
\begin{align}
\E \left[d_{NL}(U^{NL},\hat{U}^{NL})\right]&=\sum_{(x_{e,1},y_{e,1})\in\Xc_e\times\Yc_e}\E\left[d_L(U^{L},\hat{U}^{L})\left|(X_{e,1},Y_{e,1})=(x_{e,1},y_{e,1})\right.\right] 
\times \nonumber\\
&\quad\quad\quad{1\over N}\sum\limits_{\ell=1}^N \P((\Xv_{e,1}(\ell),\Yv_{e,1}(\ell))=(x_{e,1},y_{e,1}))\nonumber\\
&=\sum_{(x_{e,1},y_{e,1})\in\Xc_e\times\Yc_e} \E\left[d_L(U^{L},\hat{U}^{L})\left|(X_{e,1},Y_{e,1})=(x_{e,1},y_{e,1})\right.\right]\times\nonumber\\
&\hspace{2 cm} \E[ \pi(x_{e,1},y_{e,1}|\Xb_{e,1},\Yb_{e,1})].\label{eq:layer_t1}
\end{align}

Equations  \eqref{eq:single_layer_t1} and \eqref{eq:layer_t1} differ only in  their distributions on $\Xc_e\times\Yc_e$. Since each conditional expectation is finite (in particular, all are bounded by $d_{\max}$), we can replace channel $(\Xc_e,p(y_e|x_e),\Yc_e)$ by a bit pipe of capacity $C_e$ at time $t=1$, if we can find a coding scheme  across the layers of the stack for which, 
\begin{align}
\left|\P\left((X_{e,1},Y_{e,1})=(x_{e,1},y_{e,1})\right)-\E[ \pi(x_{e,1},y_{e,1}|\Xb_{e,1},\Yb_{e,1})]\right|,
\end{align}
can be made arbitrary small, for all ${(x_{e,1},y_{e,1})\in\Xc_e\times\Yc_e}$.

To prove that this is possible, consider a channel with input  drawn i.i.d. from some distribution $p(x_{e,1})$. We wish to build an emulation code with an  encoder that maps   $N$ source symbols, $\Xv_{e,1}\in\Xc_e^N$, to a message of $NR$ bits and a decoder that converts these $NR$ bits into a reconstruction block $\Yb_{e,1}\in\Yc_e^N$. We aim to use this  code to emulate  the DMC with transition probabilities  $\{p(y_{e,1}|x_{e,1})\}_{(x_{e,1},y_{e,1})\in\Xc_e\times\Yc_e}$ when the channel input is an i.i.d. process drawn according to   $p(x_{e,1})$. The codebook, $\Cc^{(N)}$, of this emulation code consists of $2^{NR}$ codewords, $\{\Yv_{e,1}[1],\Yv_{e,1}[2],\ldots,\Yv_{e,1}[2^{NR}]\}$, each drawn independently i.i.d.~according to $p(y_{e,1})=\sum_{x_{e,1}\in\Xc_e}p(x_{e,1})p(y_{e,1}|x_{e,1})$.  The encoder assigns message $M\in\{1,\ldots,2^{NR}\}$  to input  sequence $\Xv_{e,1}$, if $(\Xv_{e,1},\Yv_{e,1}[M])\in\Tc^{(N)}_{\e}(X_{e,1},Y_{e,1})$. If there are multiple such messages in the codebook, the encoder chooses the one with the smallest index. If there exist no  codewords in  $\Cc^{(N)}$ that are jointly typical with $\Xv_{e,1}$, then the encoder assigns message $M=1$ to $\Xv_{e,1}$.  After receiving message $M$, the decoder outputs $\Yv_{e,1}[M]$. Let $\{\pi(x_{e,1},y_{e,1}|\Xb_{e,1},\Yb_{e,1})\}_{(x_{e,1},y_{e,1})\in\Xc_e\times\Yc_e}$ be the the joint empirical  distribution between the channel input and channel output induced by running the emulation code across the $N$ copies of the bit pipe at time $t=1$. In \cite{CuffP:09}, it is shown that, the described  code can emulate channel $(\Xc_e,p(y_{e,1}|x_{e,1}),\Yc_e)$  by a  bit pipe of rate $R$, provided that $R>I(X_{e,1};Y_{e,1})$. The given emulation ensures that  the total variation between $\pi(x_{e,1},y_{e,1}|\Xb_{e,1},\Yb_{e,1})$ and $p(x_{e,1},y_{e,1})=p(x_{e,1})p(y_{e,1}|x_{e,1})$ can be made arbitrarily small as the blocklength $N$ grows without bound. In other words, there exists a sequence of codes over the bit pipe such that
\begin{align}
\left\|\boldsymbol\pi-{\bf p}\right\|_{\rm TV}\stackrel{N\to\infty}{\longrightarrow} 0,\label{eq:conv_as}
\end{align}
almost surely. (Here $\boldsymbol\pi$ and $\bf p$ are vectors describing distributions ($\pi(x_{e,1},y_{e,1}|\Xb_{e,1},\Yb_{e,1}): (x_{e,1},y_{e,1})\in\Xc_e\times\Yc_e$) and ($p(x_{e,1},y_{e,1}): (x_{e,1},y_{e,1})\in\Xc_e\times\Yc_e$) respectively.) Although Theorem 3 in \cite{CuffP:09} only guarantees convergence of $\boldsymbol\pi$ to ${\bf p}$ in probability, we can also prove almost sure convergence of  $\boldsymbol\pi$ to ${\bf p}$ using Borel-Cantelli Lemma. Let $\gamma=R-I(X_{e,1};Y_{e,1})$. Let $\Yv_{e,1}(\Xv_{e,1})$ denote the codeword in $\Cc^{(N)}$ that is assigned to $\Xv_{e,1}$ by the emulation encoder. For $\e>0$, define the error event
\[
\Ec^{(N)}=\{(\Xv_{e,1},\Yv_{e,1}(\Xv_{e,1})):\left\|\boldsymbol\pi-{\bf p}\right\|_{\rm TV}>\e\}.
\] 
Breaking the error event into two parts and then applying the union bound, Hoeffding's inequality, and the joint typicality lemma from \cite{ElGamalK_book} gives
\begin{align}
\P(\Ec^{(N)})&\leq \P(\Xv_{e,1}\notin \Tc_{\e}^{(N)}(X_{e,1}))+\left[\P((\Xv_{e,1},\Yv_{e,1}[1])\notin\Tc_{\e}^{(N)}(X_{e,1},Y_{e,1}))\right]^{2^{NR}}\nonumber\\
&\leq \sum_{x_e\in\Xc_e}\P(|\pi(x_e|\Xv_{e,1})-p(x_e)|>p(x_e)\e)+ e^{-2^{N(\gamma-\d(\e))}}\nonumber\\
&\leq  \sum_{x_e\in\Xc_e} 2e^{-2N\e^2p^2(x_e)}+ e^{-2^{N(\gamma-\d(\e))}}\nonumber\\
&\leq   2|\Xc_e|e^{-2N\e^2\min\limits_{x_e\in\Xc_e}p^2(x_e)}+ e^{-2^{N(\gamma-\d(\e))}},
\end{align}
where $\d(\e)=\e(H(Y_{e,1})+H(Y_{e,1}|X_{e,1}))\to 0$, as $\e\to 0$. Therefore, 
\[
\sum_{N=1}^{\infty}\P(\Ec^{(N)})<\infty,
\]
and hence \eqref{eq:conv_as} holds almost surely, by the Borel-Cantelli Lemma.

We next combine the emulation code with the code for $\underline{\Nc}$. The code emulates channel $p(y_e|x_e)$ at time $t=1$ across the $N$ layers of  stacked network $\underline{\Nc}_b'$ that replaces $p(y_e|x_e)$ by a link of capacity $R>C$, only at time $t=1$. The given code for $\underline{\Nc}$ can be run across $\underline{\Nc}_b'$ with expected distortion bounded as
\begin{align*}
&\E \left[d_{NL}(U^{NL},\hat{U}^{NL})\right] \nonumber\\
&= \sum_{(x_e,y_e)\in\Xc_e\times\Yc_e} \E\left[d_L(U^{L},\hat{U}^{L})\left|(X_{e,1},Y_{e,1})=(x_e,y_e)\right.\right]\times\nonumber\\
&\hspace{2 cm} \E[ \pi(x_e,y_e|\Xb_{e,1},\Yb_{e,1})]\nonumber\\
&\leq \sum_{(x_e,y_e)\in\Xc_e\times\Yc_e} \E\left[d_L(U^{L},\hat{U}^{L})\left|(X_{e,1},Y_{e,1})=(x_e,y_e)\right.\right](p(x_e,y_e)+\e)\nonumber\\
&\leq \E[d_L(U^L,\Uh^L)]+\e d_{\max}.\nonumber
\end{align*}

Thus we can replace the noisy link by a bit-pipe at time $t=1$. We use induction to extend this result to the next $n-1$ time steps. Note that in the original network
\begin{align}
&\E[d_L(U^{L},\hat{U}^{L})] =\nonumber\\
& \sum_{(x_e^n,y_e^n)\in\Xc_e^n\times\Yc_e^n} \E\left[d_L(U^{L},\hat{U}^{ L})\left|  (X_e^n,Y_e^n)=(x_e^n,y_e^n)\right.\right]\nonumber\\
&\hspace{1cm}\times\P\left((X_e^n,Y_e^n)=(x_e^n,y_e^n)\right).
\end{align}

On the other hand, using the same analysis used in deriving \eqref{eq:layer_t1}, in the $N$-fold stacked network,
\begin{align}
&\E \left[d_{NL}(U^{NL},\hat{U}^{NL})\right]\nonumber\\
& = \sum_{ (x_e^n,y_e^n)\in\Xc_e^n\times\Yc_e^n} \E\left[d_L(U^{L},\hat{U}^{L})\left| (X_e^n,Y_e^n)=(x_e^n,y_e^n)\right.\right]\nonumber\\
&\hspace{1.5cm}\times \E\left[ \pi(x_e^n,y_e^n|\Xb_e^n,\Yb_e^n)\right].\label{eq:layer_tN}
%&\hspace{1cm}\times \E\left[ \frac{\left|\left\{\ell:(\Xb^n(\ell),\Yb^n(\ell))=(x^n,y^n)\right\}\right|}{N}\right],\label{eq:layer_tN}
\end{align}
Here $\Xb_e^n=(\Xb_{e,1},\Xb_{e,2},\ldots,\Xb_{e,n})$ and $\Yb_e^n=(\Yb_{e,1},\Yb_{e,2},\ldots,\Yb_{e,n})$ refer to the inputs and outputs of  channel $e$ in the $N$ layers of the stacked network, for times $t=1,2,\ldots,n$, while $\Xb_e^n(\ell)$ and $\Yb_e^n(\ell)$ correspond to the inputs and outputs  of the emulated channel at layer $\ell$ for times $t=1,2,\ldots,n$, and 
\[
\pi(x_e^n,y_e^n|\Xb_e^n,\Yb_e^n)=\frac{\left|\left\{\ell:(\Xb_e^n(\ell),\Yb_e^n(\ell))=(x_e^n,y_e^n)\right\}\right|}{N}.
\]
Therefore, we need to show that by appropriate coding over the bit-pipes,
\begin{align}
&\left|\P\left((X_e^n,Y_e^n)=(x_e^n,y_e^n)\right)-\pi(x_e^n,y_e^n|\Xb_e^n,\Yb_e^n)\right|\label{eq:p_e}
\end{align}
can be made arbitrarily small.  Note that
\begin{align}
&\P\left((X_e^n,Y_e^n)=(x_e^n,y_e^n)\right)=\nonumber\\
& \prod\limits_{t=1}^n\P\left((X_{e,t},Y_{e,t})=(x_{e,t},y_{e,t})\left|(X_e^{t-1},Y_e^{t-1})=(x_e^{t-1},y_e^{t-1})\right.\right),\label{eq:single_layer_t_all}
\end{align}
and 
\begin{align} 
&\pi(x_e^n,y_e^n|\Xb_e^n,\Yb_e^n)=\prod\limits_{t=1}^n\frac{\pi(x_e^t,y_e^t|\Xb_e^t,\Yb_e^t)}
{\pi(x_e^{t-1},y_e^{t-1}|\Xb_e^{t-1},\Yb_e^{t-1})},\label{eq:layer_t_all}
\end{align}
where for $t=1$ 
\[
\pi(x_e^{t-1},y_e^{t-1}|\Xb_e^{t-1},\Yb_e^{t-1})=1.
\]
We have already proven that we can make the first term in the product in \eqref{eq:layer_t_all}  converge to the first term in  the product in \eqref{eq:single_layer_t_all}  with probability one. We next prove by induction that the same result is true for each subsequent term in \eqref{eq:single_layer_t_all} and  \eqref{eq:layer_t_all}. Since all of the terms in \eqref{eq:layer_t_all} are positive and  upper-bounded by $1$, so too is their product. Thus,  the Dominated Convergence Theorem (see, for example, \cite{Durrett:96})  shows that \eqref{eq:p_e} can be made arbitrarily small provided that each term converges almost surely.

To apply induction, assume that there exist $t-1$ emulation codes whose application makes the first $t-1$ terms in \eqref{eq:layer_t_all} each converge to the corresponding term in \eqref{eq:single_layer_t_all}  almost surely.  Using this inductive hypothesis, we prove that the  $t^{\rm th}$ term in \eqref{eq:layer_t_all} converges to the $t^{\rm th}$ term in \eqref{eq:single_layer_t_all} as well.

Given the inductive hypothesis that 
\begin{align}
\frac{\pi(x_e^{t'},y_e^{t'}|\Xb_e^{t'},\Yb_e^{t'})}{\pi(x_e^{t'-1},y_e^{t'-1}|\Xb_e^{t'-1},\Yb_e^{t'-1})}\to p(x_{e,t'},y_{e,t'}|x_e^{t'-1},y_e^{t'-1})
\end{align}
almost surely, for all $(x_e^{t'},y_e^{t'})$ and all $t'\leq t-1$, it follows  that 
\begin{align}
\pi(x_e^{t-1},y_e^{t-1}|\Xb_e^{t-1},\Yb_e^{t-1})\to p(x_e^{t-1},y_e^{t-1})
\end{align}
almost surely, for all $(x^{t-1},y^{t-1})$. Since the two networks apply precisely the same deterministic code to the channel outputs at time $t-1$ to create the channel inputs  at time $t$, this bound implies
\begin{align}
\pi(x_e^{t},y_e^{t-1}|\Xb_e^{t},\Yb_e^{t-1})\to p(x_e^{t},y_e^{t-1})
\end{align}
almost surely, for all $(x^{t},y^{t-1})$ as well. We now show that if the emulation code used at time $t$ is generated independently of the codes used at times $1,2,\ldots,t-1$, then for each $(x_e^t,y_e^t)$, 
\begin{align}
\frac{\pi(x_e^t,y_e^t|\Xb_e^t,\Yb_e^t)}{\pi(x_e^{t},y_e^{t-1}|\Xb_e^{t},\Yb_e^{t-1})}\to p(y_{e,t}|x_{e,t})
\end{align}
almost surely, where $p(y_{e,t}|x_{e,t})=\P(Y_e=y_{e,t}|X_e=x_{e,t})$. Note that
\begin{align}
&\P(\Yv_{e,t}(1)=y_{e,t}|(\Xv_e^{t}(1),\Yv_e^{t-1}(1))=(x_e^t,y_e^{t-1}))\nonumber\\
&=\sum_{s^{N}\in\Xc_e^{N}:s_1=x_{e,1}} \P(\Yv_{e,t}(1)=y_{e,t},\Xv_{e,t}(2:N)=s_2^{N}|(\Xv_e^{t}(1),\Yv_e^{t-1}(1))=(x_e^t,y_e^{t-1}))\nonumber\\
&=\sum_{s^{N}\in\Xc_e^{N}:s_1=x_{e,1}} \P(\Xv_{e,t}(2:N)=s_2^{N}|(\Xv_e^{t}(1),\Yv_e^{t-1}(1))=(x_e^t,y_e^{t-1}))\nonumber\\
&\hspace{2cm}\times\P(\Yv_{e,t}(1)=y_{e,t}|\Xv_{e,t}=s^{N},(\Xv_e^{t-1}(1),\Yv_e^{t-1}(1))=(x_e^t,y_e^{t-1}))\nonumber\\
&=\sum_{s^{N}\in\Xc_e^{N}:s_1=x_{e,1}} \P(\Xv_{e,t}(2:N)=s_2^{N}|(\Xv_e^{t}(1),\Yv_e^{t-1}(1))=(x_e^t,y_e^{t-1}))\nonumber\\
&\hspace{2cm}\times\P(\Yv_{e,t}(1)=y_{e,t}|\Xv_{e,t}=s^{N}),\label{eq:limit0}
\end{align}
where the last equality holds because $(\Xv_e^{t-1},\Yv_e^{t-1})\to \Xv_{e,t} \to \Yv_{e,t}$ since the emulation code maps $\Xv_{e,t}$ to $\Yv_{e,t}$ independently of all prior channel inputs and outputs.

Since each network layer independently operates an identical code, and codewords in the emulation codebook are generated according to an i.i.d.~distribution, it follows that
\begin{align*}
\P(\Yv_{e,t}(1)=y_{e,t}|\Xv_{e,t}=s^{N})=\P(\Yv_{e,t}(\ell)=y_{e,t}|\Xv_{e,t}=s^{N})
\end{align*}
for any $\ell$ such that $s_{\ell}=x_{e,t}$ under the operation of a random emulation code.
Therefore,
\begin{align}
&\P(\Yv_{e,t}(1)=y_{e,t}|\Xv_{e,t}=s^{N})\nonumber\\
&={1\over N\pi(x_{e,t}|s^N)}\sum_{\ell: s_{\ell}=x_{e,t}}\P(\Yv_{e,t}(\ell)=y_{e,t}|\Xv_{e,t}=s^{N})\nonumber\\
&={1\over N\pi(x_{e,t}|s^N)}\sum_{\ell: s_{\ell}=x_{e,t}}\E[\ind_{\Yv_{e,t}(\ell)=y_{e,t}}|\Xv_{e,t}=s^{N}]\nonumber\\
&=\E\left[\left.{1\over N\pi(x_{e,t}|s^N)}\sum_{\ell: s_{\ell}=x_{e,t}}\ind_{\Yv_{e,t}(\ell)=y_{e,t}}\right|\Xv_{e,t}=s^{N}\right]\nonumber\\
%&=\E\left[\left.{1\over N\pi(x_{e,t}|s^N)}\sum_{\ell: s_{\ell}=x_{e,t}}\ind_{\Yv_{e,t}(\ell)=y_{e,t}}\right|\Xv_{e,t}=s^{N}\right]\nonumber\\
&=\E\left[\left.{\pi(x_{e,t},y_{e,t}|\Xv_{e,t},\Yv_{e,t})\over \pi(x_{e,t}|\Xv_{e,t})}\right|\Xv_{e,t}=s^{N}\right].\label{eq:limit1}
\end{align}
By our inductive assumption and an argument similar to the one used in Remark 1, if $s^N\in\Tc_{\e}^{(N)}(X_{e,t})$, for $N$ large enough
\begin{align}
&\left|\E\left[\left.{\pi(x_{e,t},y_{e,t}|\Xv_{e,t},\Yv_{e,t})\over \pi(x_{e,t}|\Xv_{e,t})}\right|\Xv_{e,t}=s^N\right]-p(y_{e,t}|x_{e,t})\right|\nonumber\\
&<\e.\label{eq:limit2}
\end{align}

Combining \eqref{eq:limit0}, \eqref{eq:limit1} and \eqref{eq:limit2}, it follows that
\begin{align}
&\P(\Yv_{e,t}(1)=y_{e,t}|(\Xv_e^{t}(1),\Yv_e^{t-1}(1))=(x_e^t,y_e^{t-1}))\nonumber\\
&=\sum_{s^N\in\Tc_{\e}^{(N)}(X_{e,t}):s_1=x_{e,1}} \P(\Xv_{e,t}(2:N)=s_2^{N}|(\Xv_e^{t}(1),\Yv_e^{t-1}(1))=(x_e^t,y_e^{t-1}))\nonumber\\
&\hspace{2cm}\times\P(\Yv_{e,t}(1)=y_{e,t}|\Xv_{e,t}=s^{N})\nonumber\\
&\hspace{0.5cm}+\sum_{s^N\notin\Tc_{\e}^{(N)}(X_{e,t}):s_1=x_{e,1}} \P(\Xv_{e,t}(2:N)=s_2^{N}|(\Xv_e^{t}(1),\Yv_e^{t-1}(1))=(x_e^t,y_e^{t-1}))\nonumber\\
&\hspace{2cm}\times\P(\Yv_{e,t}(1)=y_{e,t}|\Xv_{e,t}=s^{N})\nonumber\\
&\;\;\; \leq (p(y_{e,t}|x_{e,t}) + \e)\P\left(\Xv_{e,t}\in\Tc_{\e}^{(N)}(X_{e,t}) |(\Xv_e^{t}(1),\Yv_e^{t-1}(1))=(x_e^t,y_e^{t-1})\right)\nonumber\\
&\hspace{0.5cm} +\P(\Xv_{e,t}\notin \Tc_{\e}^{(N)}(X_{e,t})|(\Xv_e^{t}(1),\Yv_e^{t-1}(1))=(x_e^t,y_e^{t-1})).\label{eq:34}
\end{align}
Similarly, 
\begin{align}
&\P(\Yv_{e,t}(1)=y_{e,t}|(\Xv_e^{t}(1),\Yv_e^{t-1}(1))=(x_e^t,y_e^{t-1}))\nonumber\\
&\;\;\; \geq (p(y_{e,t}|x_{e,t}) - \e)\P\left(\Xv_{e,t}\in\Tc_{\e}^{(N)}(X_{e,t}) |(\Xv_e^{t}(1),\Yv_e^{t-1}(1))=(x_e^t,y_e^{t-1})\right).\label{eq:35}
\end{align}
But, if $\P((\Xv_e^{t}(1),\Yv_e^{t-1}(1))=(x_e^t,y_e^{t-1}))\neq 0$, then 
\begin{align}
&\P\left(\Xv_{e,t}\notin \Tc_{\e}^{(N)}(X_{e,t})|(\Xv_e^{t}(1),\Yv_e^{t-1}(1))=(x_e^t,y_e^{t-1})\right)\nonumber\\
&\;\;\; ={\P\left(\Xv_{e,t}\notin \Tc_{\e}^{(N)}(X_{e,t}),(\Xv_e^{t}(1),\Yv_e^{t-1}(1))=(x_e^t,y_e^{t-1})\right) \over \P\left((\Xv_e^{t}(1),\Yv_e^{t-1}(1))=(x_e^t,y_e^{t-1})\right)} \nonumber\\
&\;\;\; \leq {\P\left(\Xv_{e,t}\notin \Tc_{\e}^{(N)}(X_{e,t})\right) \over \P\left((\Xv_e^{t}(1),\Yv_e^{t-1}(1))=(x_e^t,y_e^{t-1})\right)}\to 0\label{eq:36}
\end{align}
as $N\to\infty$, and hence $\P(\Xv_{e,t}\in\Tc_{\e}^{(N)}(X_{e,t}) |(\Xv_e^{t}(1),\Yv_e^{t-1}(1))=(x_e^t,y_e^{t-1}))\to 1$, as $N\to\infty$. Therefore, combining \eqref{eq:34}, \eqref{eq:35}, and \eqref{eq:36}, it follows that, for each $(x_e^t,y_e^t)$, 
\begin{align}
&\P(\Yv_{e,t}(1)=y_{e,t}|(\Xv_e^{t}(1),\Yv_e^{t-1}(1))=(x_e^t,y_e^{t-1}))\to p(y_{e,t}|x_{e,t}),
\end{align}
almost surely, as $N$ grows to infinity.
 
This concludes the proof, because it shows that, for each $\ell\in\{1,2,\ldots,N\}$, as the number of layers $N$ grows, $\Yv_{e,t}(\ell)$ becomes independent of $(\Xv_e^{t-1}(\ell),\Yv_e^{t-1}(\ell))$ conditioned on  $\Xv_{e,t}(\ell)$, and its conditional distribution converges to $p(y_{e,t}|x_{e,t})$ corresponding to the transition probability of channel $e$. 
\end{itemize}
\end{proof}

\begin{remark}
The first part of the proof of Theorem \ref{thm:main} is not specific to DMCs, and shows that $\Dc(\kappa,\Nc_b)\subseteq\Dc(\kappa,\Nc)$ for all networks $\Nc$ of (discrete or continuous) point-to-point channels.
\end{remark}

%******************************************************************************
%******************************************************************************

\section{Continuity: zero-distortion versus lossless}\label{sec:zero-dist}

\begin{figure}[t]
\begin{center}
\psfrag{U}[l]{\footnotesize $U^L$}
\psfrag{Uh}[l]{\footnotesize $\Uh^L$}
\psfrag{X}[l]{\footnotesize $X^n$}
\psfrag{Y}[l]{\footnotesize $Y^n$}
\psfrag{R}[l]{\footnotesize $LR$}
%\psfrag{Encoder}[l]{\footnotesize{Enc.}}
%\psfrag{Decoder}[l]{\footnotesize{Dec.}}
\includegraphics[width=10cm]{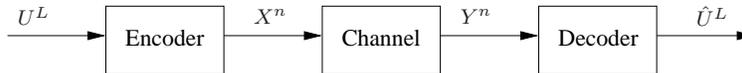}\caption{Simple point-to-point channel}\label{fig:P2P}
\end{center}
\end{figure}
The distortion criteria for lossless source coding and lossy source coding with a distortion constraint of zero are different. In lossless coding, we require that the probability of error in reconstructing a vector of source symbols goes to zero as the blocklength of that vector grows without bound. In lossy  coding, we require that the per symbol distortion between the source vector and its reconstruction approach zero for sufficiently long blocklengths. As a result, even under the Hamming distortion measure, distortion 0 reconstructions do not necessarily meet the lossless source reconstruction criterion.  Before investigating the relationship between these problems in a generic network $\Nc$ of the form defined in Section \ref{sec:setup}, we consider some special cases where the relationship is known.  
Consider the simple point-to-point network shown in Fig.~\ref{fig:P2P}. Let the source $U$ be i.i.d. and distributed according to $p(u)$, and let $C=\max_{p(x)}I(X;Y)$ denote the capacity of the point-to-point channel connecting the source and the  destination.   The minimal required rate for describing the source $U$ at distortion $D$ is \cite{cover} 
$R(D)=\min_{p(\hat{u}|u):\E[d(U,\Uh)]\leq D}I(U;\Uh).$ In such point-to-point networks separation of source coding and channel coding is known to be optimal \cite{Shannon:48}. Hence to describe the source at distortion $D$, we need $C\geq \kappa R(D)$.  Evaluating $R(D)$ at $D=0$ gives
\[
R(0)=\min\limits_{p(\hat{u}|u):\E[d(U,\Uh)]=0}I(U;\Uh)=I(U;U)=H(U),
\]
where $H(U)$ is the entropy rate of the source $U$. Since the minimal  rate for lossless reconstruction of  the source $U$ is also the entropy rate, the zero-distortion and lossless reconstruction rate regions coincide in this simple network. Explicit characterizations of the multi-dimensional rate-distortion regions for general multiuser networks are unknown. Therefore, proving or disproving the equivalence of zero-distortion and lossless reconstruction rate-regions in such networks requires more elaborate analysis. In his Ph.D. thesis, W.H.~Gu proved that in noiseless networks consisting of  point-to-point bit-pipes,  zero-distortion and lossless reconstruction rate regions coincide \cite{Gu_thesis}.

In this section, we prove the equivalence of zero-distortion reconstruction  and lossless reconstruction in general networks described by multiuser discrete memoryless channels (mDMCs) with statistically dependent sources. More precisely, we prove that  in any mDMC with independent or dependent sources, lossless reconstruction is achievable if and only if zero-distortion reconstruction is achievable. 

Consider network $\Nc$  shown in Fig.~\ref{fig:mDMC}, which consists of a general mDMC described by  \[p(y^{(1)},\ldots,y^{(m)}|x^{(1)},\ldots,x^{(m)}).\] Let $\Vc\triangleq\{1,\ldots,m\}$. Node $a\in\Vc$ observes source process $\Uv^{(a)}$ and is interested in reconstructing sources observed by the other nodes. The coding operations are very similar to the  case of wired networks. Each node observes a  block of length $L$ of its own source symbols and describes them to the other nodes in $n$ channel uses. As before, the coding rate  $\kappa$ is defined as $\kappa={L/n}$. At each each time $t=1,\ldots,n$, node $a$  generates channel input  $X^{(a)}_{t}$ as a function its own source block $U^{(a),L}$ and its received channel outputs up to time $t-1$, \ie $Y^{(a),t-1}$. In other words, $X^{(a)}_{t}=X^{(a)}_{t}(U^{(a),L},Y^{(a),t-1})$. The set of achievable distortion matrices on network $\Nc$ at rate $\kappa$ is denoted by $\Dc(\kappa,\Nc)$.  Throughout this section we assume that for any $(a,b)\in\Vc^2$, $d^{(a\to b)}(u,\uh)=0$ if and only if $u=\uh$.

Given any $D\in\Dc(\kappa,\Nc)$, let 
\[
\Lc(D)\triangleq\{(a,b):D(a,b)=0\}.
\]

\begin{figure}
\begin{center}
\psfrag{p}{\footnotesize $p(y^{(1)},\ldots,y^{(m)}|x^{(1)},\ldots,x^{(m)})$}
\psfrag{x}{\footnotesize $\to X^{(a)} $}
\psfrag{y}{\footnotesize $\leftarrow Y^{(a)}$}
\includegraphics[width=6.2cm]{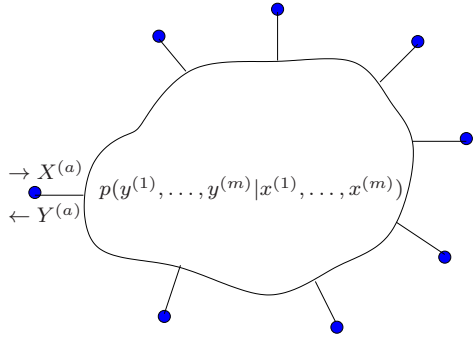} \caption{General multiuser discrete memoryless channel (mDMC)}\label{fig:mDMC}
\end{center}
\end{figure}

\begin{theorem}\label{thm:lossless_zero_dist}
Fix any non-negative matrix $D=(D(a,b):(a,b)\in\Vc^2)$ with $|\Lc(D)|>0$. For any $(a,b)\in\Lc(D)$, assume that $H(U^{(a)}|(U^{(c)})_{c\in\Vc\backslash a})>0$. Then  $D\in\Dc(\kappa,\Nc)$ if and only if, for any $\e>0$ there exists integers $L$ and $n\leq L/(\kappa-\e)$, for which we can design a code of source blocklength $L$ and channel blocklength  $n$ that satisfies 
\[
 \P(U^{(a),L}\neq \Uh^{(a\to b),L})\leq \e,
 \] 
 for all $(a,b)\in\Lc(D)$ and 
 \[
 \E[d_L(U^{(a),L},\Uh^{(a\to b),L})]\leq D(a,b)+\e,
 \] 
 for all $(a,b)\in\Vc^2\backslash\Lc(D)$. 
 \end{theorem}

\begin{proof}[Proof of Theorem \ref{thm:lossless_zero_dist}]

For the forward result, fix a sequence of codes at rate $L/n\to \kappa$, distortion $\E[d_L(U^{(a),L},\Uh^{(a\to b),L})]\to D(a,b)$ for all $(a,b)\notin \Lc(D)$ and error probability  $\P(U^{(a),L}\neq \hat{U}^{(a\to b),L})\to 0$ for all $(a,b)\in\Lc(D)$. For each $(a,b)\in\Lc(D)$, the given sequence of codes satisfies 
\begin{align*}
\E&[d_L(U^{(a),L},\hat{U}^{(a\to b),L})] \nonumber\\
&=\E\left[d_L(U^{(a),L},\hat{U}^{(a\to b),L})\left|U^{(a),L}\neq \hat{U}^{(a\to b),L}\right.\right]\P(U^{(a),L}\neq \hat{U}^{(a\to b),L})\nonumber\\
&+\E\left[d_L(U^{(a),L},\hat{U}^{(a\to b),L})\left|U^{(a),L} =      \hat{U}^{(a\to b),L}\right.\right]\P(U^{(a),L}  =     \hat{U}^{(a\to b),L})\nonumber\\
&\leq d_{\max}\P(U^{(a),L}\neq \hat{U}^{(a\to b),L}).
\end{align*}
Since the given bound approaches 0 as $\P(U^{(a),L}\neq \hat{U}^{(a\to b),L})\to 0$, the sequence of codes  achieves zero-distortion reconstruction of source $a$ at node $b$, which is the the desired result.

To prove the converse, fix any $D\in\Dc(\kappa,\Nc)$ with $|\Lc(D)|>0$ and  any $\epsilon>0$. By the  definition of $\Dc(\kappa,\Nc)$, for any $\e>0$, there exists a  code with source blocklength  $L$ and channel blocklength $n\leq L/( \kappa -\e)$ such that 
\begin{align}
\E[d_L(U^{(a),L},\hat{U}^{(a\to b),L})]\leq D(a,b)+\e\label{eq:code_e}
\end{align}
for each $(a,b)\in\Vc^2$. Specifically, for any $(a,b)$ such that $D(a,b)=0$, 
\begin{align*}
\E[d_L(U^{(a),L},\hat{U}^{(a\to b),L})]\leq \e.
\end{align*}

We now prove that with an asymptotically negligible  increase in  number of channel uses $n$, node $a$ can send node $b$ sufficient information to improve node $b$'s reconstruction of node $a$'s data from a  zero-distortion reproduction to a lossless reconstruction. We further show that this change preserves the quality of all other reconstructions.  

The following argument builds a code of source blocklength $NL$ and channel blocklength $n(N+N')$, for some integer $N'$ to be defined shortly, from the given code of source blocklength $L$ and channel blocklength $n$.

Each node $a\in\Vc$ breaks its incoming source block of length  $NL$  into $N$ non-overlapping blocks of length $L$, given by
\[
U^{(a),L},U_{L+1}^{(a),2L},\ldots,U_{(N-1)L+1}^{(a),NL}.
\]
Each node then applies the  blocklength-$L$ code   $N$ times to independently code each of these blocks. In total, this requires $Nn$ channel uses. Independently decoding each $L$-block with the blocklength-$L$ decoder achieves, for each $a,b\in\Vc$,  a reconstruction of length $NL$ such that
\begin{align}
\E[d_L(U_{(\ell-1)L+1}^{(a),\ell L},\hat{U}_{(\ell-1)L+1}^{(a\to b),\ell L})]\leq D(a,b)+\e,\label{eq:code_e_N}
\end{align}
for each $\ell=1,2,\ldots,N$. 

For  $(a,b)\in\Lc(D)$ and each $\ell\in\{1,\ldots,N\}$, denote the input of node $a$ in session $\ell$ as \[U^{L}(\ell)\triangleq U^{(a),\ell L}_{(\ell-1)L+1},\] and the corresponding output at node $b$ as \[\hat{U}^L(\ell)\triangleq\hat{U}^{(a\to b),\ell L}_{(\ell-1)L+1}.\] By assumption,
\[
\E[d_L(U^L(\ell),\Uh^L(\ell))]\leq \e.
\]
Thus
\begin{align}
\e&\geq \E[d_L(U^L(\ell),\Uh^L(\ell))]\nonumber\\
& ={1\over L}\sum_{i=1}^L \E [d(U_i(\ell),\Uh_i(\ell))]\nonumber\\
&  \geq {1\over L} \sum_{i=1}^L d_{\min}\P(U_i(\ell)\neq \Uh_i(\ell)),\label{eq:bit_p_error}
\end{align}
where $d_{\min}\triangleq\min_{(u,\uh)\in\Uc\times\hat{\Uc}:u\neq\uh} d(u,\uh)$. Since  all alphabets are assumed to be finite, and $d(u,\uh)=0$ if and only if $u=\uh$, $d_{\min}>0$ by assumption. Therefore, 
\[
{1\over L}\sum_{i=1}^L \P(U_i(\ell)\neq \Uh_i(\ell)) \leq {\e \over d_{\min}}
\]
for all $\ell\in\{1,2,\ldots,N\}$.

\begin{figure}[t]
\begin{center}
\psfrag{U1}[l]{$U^L(1)$}
\psfrag{U2}[l]{$U^L(2)$}
\psfrag{U3}[l]{$U^L(N)$}
\psfrag{Uh1}[l]{$\Uh^L(1)$}
\psfrag{Uh2}[l]{$\Uh^L(2)$}
\psfrag{Uh3}[l]{$\Uh^L(N)$}
\psfrag{Layer 1}[l]{Session 1}
\psfrag{Layer 2}[l]{Session 2}
\psfrag{Layer N}[l]{Session N}
\includegraphics[width=7cm]{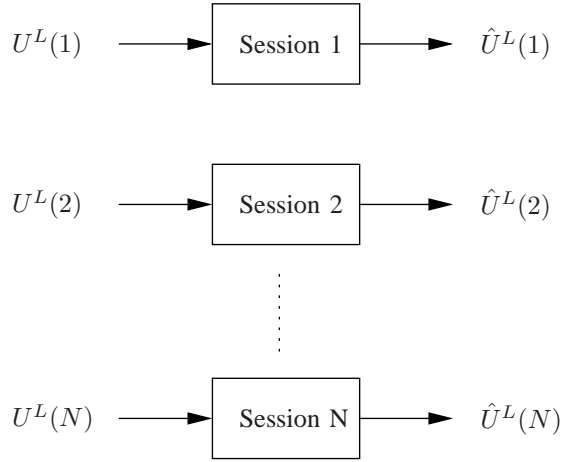}\caption{Source $U$ and its reconstruction at $L$ parallel sessions.}\label{fig:layers}
\end{center}
\end{figure}

\begin{figure}[h]
\begin{center}
\psfrag{Enc.}[l]{\footnotesize Encoder}
\psfrag{Dec.}[l]{\footnotesize Decoder}
\psfrag{U1}[l]{$U^{L}$}
\psfrag{Uh1}[l]{$\Uh^{L}$}
\psfrag{Ut1}[l]{$U^{L}$}
\psfrag{R0}[l]{$R_0$}
\includegraphics[width=8cm]{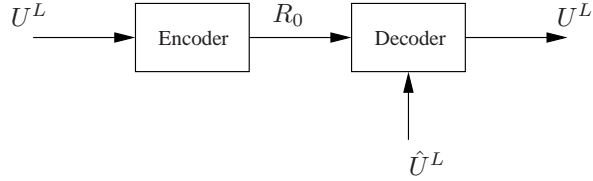}\caption{Slepian-Wolf coding for converting zero-distortion reconstruction into lossless reconstruction.}\label{fig:SW}
\end{center}
\end{figure}

Recall that all  sources and channels are memoryless by assumption and that the same code is used independently on each $L$-vector. Therefore, $\{U^{L}(\ell),\Uh^L(\ell)\}_{\ell=1}^N$  is an i.i.d. sequence. (See Fig.~\ref{fig:layers}.)  Our goal in the argument that follows is to losslessly describe $U^L(1),\ldots,U^L(N)$ to a decoder that knows $\Uh^L(1),\ldots,\Uh^L(N)$. We treat this as a problem of lossless source coding with receiver side information, as shown in Fig.~\ref{fig:SW}. From \cite{SlepianW:73},  rate $R_0^{(k)}=H(U^L|\Uh^L)$ suffices for losslessly reconstructing $U^L$ at a receiver that knows $\hat{U}^L$. Here lossless coding means that  the reconstruction  $\Ut^{LN}$ at the receiver has an error probability    $\P(U^{LN}\neq \tilde{U}^{LN})$ that can be made arbitrarily small, which is precisely the criterion needed for our proof.  Therefore, for any $\e<d_{\min}/2$, using Fano's inequality \cite{cover}, Jensen's inequality, and the concavity of the entropy function, we have
\begin{align}
R_0^{(k)}&=H(U^L|\Uh^L)=\sum_{i=1}^L H(U_i|U^{i-1},\Uh^L)\nonumber\\
&\leq \sum_{i=1}^L H(U_i|\Uh_i)\nonumber\\
&\leq \sum_{i=1}^L H(U_i,\ind_{U_i=\Uh_i}|\Uh_i)\nonumber\\
%&\leq \sum_{i=1}^L[H(\ind_{U_i=\Uh_i}|\Uh_i)+ H(U_i|\Uh_i,\ind_{U_i=\Uh_i})]\nonumber\\
%&\leq \sum_{i=1}^L[H(\ind_{U_i=\Uh_i})+ |\Uc|\P(U_i\neq\Uh_i)]\nonumber\\
&\leq \sum_{i=1}^L[h(\P(U_i\neq \Uh_i))+\log|\Uc|\P(U_i\neq\Uh_i)]\nonumber\\
&\leq L h\left({1\over L}\sum_{i=1}^L \P(U_i\neq \Uh_i)\right)+ \log|\Uc|\sum_{i=1}^L \P(U_i\neq\Uh_i)\nonumber\\
&\leq L\left(h\left({\e\over d_{\min}}\right)+{\log|\Uc|\e\over d_{\min}}\right)\nonumber\\
&\triangleq Lf(\e),\label{eq:upper_bd_on_cond_entropy}
\end{align}
 where for any $0\leq p\leq 1$, $h(p)=-p\log p-(1-p)\log(1-p)$, and $f(\e)\triangleq h({\e\over d_{\min}})+{\log|\Uc|\e\over d_{\min}}$. Note that $f(\e)\to 0 $ as $\e\to 0$. 

For each $(a,b)\in\Lc(D)$, we send the rate-$R_0^{(k)}$ description of $U^L$ from node $a$ to node $b$ by treating the random mapping from $U^L$ to $\Uh^L$ that results from applying the given code across the given network as a noisy channel. Specifically, we order the source-receiver pairs $(a,b)\in\Lc(D)$ lexicographically and send the description for  the $k$-th pair $(a,b)$ using  $N'_k$ dummy source vectors $U^L(N+\sum_{k'=1}^{k-1}N'_{k'}+1),\ldots,U^L(N+\sum_{k'=1}^{k-1}N'_{k'}+N'_k)$, thereby creating $N'_k$ uses of a channel $p(\uh^L|u^L)$ through which we can reliably transmit the lossless description of $U^L(1),\ldots,U^L(N)$ for $(a,b)$ to the decoder. The decoder's distortion-$\e$ reconstructions  $\Uh^L(1),\ldots,\Uh^L(N)$ of source vectors $U^L(1),\ldots,U^L(N)$ are treated as side information known only by the decoder.

The following discussion describes the approach precisely and investigates its performance. The code used to losslessly describe $U^L(1),\ldots,U^L(N)$ from node $a$ to the node $b$  employs fixed source values $U^{(v),L}(N+\sum_{k'=1}^{k-1}N'_{k'}+1)=\ldots=U^{(v),L}(N+\sum_{k'=1}^{k-1}N'_{k'}+N'_k)=u^{(v),L}$ for all  nodes $v\in\Vc\backslash a$ in the network. The value transmitted by each node $v\in\Vc\backslash \{a\}$  is chosen as follows. 

Since distortion in non-negative by assumption,
\begin{align*}
\e &\geq \E[d_L(U^L,\Uh^L)]\nonumber\\
&\geq \E[d_L(U^L,\Uh^L)|U^{(-a),L}\in\Tc^{(L)}_{\d}]\P(U^{(-a),L}\in\Tc^{(L)}_{\d}),
\end{align*}
where $U^{(-a),L}\triangleq(U^{(v),L})_{v\in\Vc \backslash a}$.
For any $\d>0$ and all  $L$ large enough, $\P(U^{(-a),L}\in\Tc^{(L)}_{\d}) >1-\d$, which implies that
\begin{align*}
\E[d_L(U^L,\Uh^L)|U^{(-a),L}\in\Tc^{(L)}_{\d}]\leq {\e\over 1-\d}.
\end{align*} 
Hence, there exists $u^{(-a),L}\in\Tc^{(L)}_{\d}$ such that
\begin{align}
\E[d_L(U^L,\Uh^L)|U^{(-a),L}=u^{(-a),L}]\leq {\e\over 1-\d}.\label{eq:fixed_blocks}
\end{align}

Fix any such $u^{(-a),L}$.  To bound the capacity of the resulting channel, we first bound the conditional entropy of $U^L$ given $\Uh^L$, when $U^{(-a),L}=u^{(-a),L}$. Here, following steps similar to those in \eqref{eq:bit_p_error} and \eqref{eq:upper_bd_on_cond_entropy}, but here conditioning on  $U^{(-a),L}=u^{(-a),L}$,  we conclude that
\begin{align*}
H(U^L|\Uh^L,U^{(-a),L}=u^{(-a),L})\leq Lf({\e\over 1-\d}).
\end{align*}
To finish our capacity calculation, we next bound the entropy of $U^L$ given $U^{(-a),L}=u^{(-a),L}$. Since $u^{(-a),L}\in\Tc^{(L)}_{\d}$, for any $u^L\in\Tc_{\d}^{(L)}(U|u^{(-a),L})$,  
\begin{align*}
p(u^L|u^{(-a),L})\leq 2^{-(1-\d)LH(U|U^{(-a)})}
\end{align*}
by \cite{ElGamalK_book}. Hence, for $L$ large enough,
\begin{align*}
&H(U^L|U^{(-a),L}=u^{(-a),L})\nonumber\\
&=\sum\limits_{u^L}-p(u^L|u^{(-a),L}) \log p(u^L|u^{(-a),L})\nonumber\\
&\geq\sum\limits_{u^L\in\Tc^{(L)}_{\d}(U|u^{(-a),L})}\hspace{-0.8cm}-p(u^L|u^{(-a),L}) \log p(u^L|u^{(-a),L})\nonumber\\
& \geq(1-\d)LH(U|U^{(-a)})\P(U^L\in\Tc_{\d}^{(L)}|U^{(-a),L}=u^{(-a),L})\nonumber\\
&\geq(1-\d)^2LH(U|U^{(-a)}),
\end{align*}
where the last line follows since, for $L$ large enough, $\P(U^L\in\Tc_{\d}^{(L)}|U^{(-a),L}=u^{(-a),L})>1-\d$.

Hence, fixing $U^{(-a),L}=u^{(-a),L}$ yields a channel $p(\uh^L|u^L,U^{(-a),L}=u^{(-a),L})$, with capacity 
\begin{align}
C_0^{(k)} \geq (1-\d)^2LH(U|U^{(-a)})-Lf({\e\over 1-\d}).\label{eq:rate0}
\end{align} 

Thus the rate required to losslessly describe $U^{LN}$ to a decoder with  reproduction $\Uh^{LN}$ of $U^{LN}$ is at most $R_0^{(k)}N$, and the capacity of the channel over which we wish to describe $U^{LN}$ is at least $C_0^{(k)}$ bits per $L$ network uses. We can therefore achieve the desired lossless description  provided that $N'_kC_0^{(k)}>NR_0^{(k)}$, giving $N'_k>NR_0^{(k)}/C_0^{(k)}$. Thus the total number of sessions required to send first the lossy description and then the lossless incremental description is 
\[
N+N'=N+\sum_{k=1}^{|\Lc(D)|}N'_k>N\Big(1+\sum_{k=1}^{|\Lc(D)|} R_0^{(k)}/C_0^{(k)}\Big).
\]
 Here 
\begin{align*}
{R_0^{(k)}\over C_0^{(k)}}&\leq {Lf(\e)\over (1-\d)^2LH(U|U^{(-a)})Lf({\e\over 1-\d}) }\nonumber\\
&={f(\e)\over (1-\d)^2H(U|U^{(-a)})-f({\e\over 1-\d})},
\end{align*}
which approaches zero as $\e$ approaches zero and $\d$ approaches zero.  Repeating this process for every $(a,b)\in\Lc(D)$, the resulting coding rate  can be bounded as
\[
{\kappa\over 1+\sum\limits_{k=1}^{|\Lc(D)|}{R_0^{(k)}/ C_0^{(k)}}} \leq \kappa' \leq\kappa.
\]
Since  $|\Lc(D)|<|\Vc|^2$ is a finite number, the resulting coding rate $\kappa'$, after adding these extra sessions, still approaches to $\kappa$,  as $\e$ and $\d$ corresponding to each $(a,b)\in\Lc(D)$ converge to zero.

\end{proof}

Combining Theorem \ref{thm:lossless_zero_dist}, Theorem \ref{thm:main} and the result proved by W.~Gu in \cite{Gu_thesis} proves the separation of source-network coding and channel coding in a wireline network with dependent sources with lossy or lossless reconstructions. In particular, this result partially extends the separation result of \cite{KoetterE:09a} to the case where the sources are dependent. The extension is partial since in \cite{KoetterE:09a} the channels can be discrete or continuous, but here we have only considered discrete channels. In the next section, we consider the case of AWGN channels.

%******************************************************************************
%******************************************************************************

\section{Continuous channels}\label{sec:awgn}

While the capacity results of \cite{KoetterE:11} are proven for general (discrete or continuous) alphabets, the sources and channels considered in Theorems \ref{thm:stack} and \ref{thm:main} were all assumed to have finite alphabets. In this section, we prove that our results also hold for AWGN channels. In order to prove this we use the discretization method introduced  in \cite{McEliece:77}.

Consider a wireline network $\Nc$  with an AWGN channel from node $a$ to node $b$. Let the input and output of this channel be $X$ and $Y=X+Z$, respectively. The coding on $\Nc$ is performed similar to the coding described in Section \ref{sec:setup}.  Assume  input power constraint $P$ and noise power $N$. To impose the power constraint, for a code with channel blocklength $n$, we require
\[
\E[X_t^2]\leq P,
\]
for  $t=1,\ldots,n$. Similarly, in the $N$-fold stacked version of $\Nc$, we require
\[
{1\over N}\sum_{\ell=1}^N\E[X_t^2(\ell)]\leq P,
\]
for  $t=1,\ldots,n$. 

Let $\Nc_b$ be a wireline network that is identical to network $\Nc$ except that   the channel from $a$ to $b$ is replaced by a bit pipe  of capacity $C=0.5\log(1+P/N)$. Theorem \ref{thm:awgn2} shows, as in the case of discrete-valued channels, that this change does not affect the set of achievable distortions, thereby generalizing Theorem \ref{thm:main}. 

\begin{remark}
Given a Gaussian channel with input power constraint $P$, usually, a code of blocklength $n$ and rate $R$  is  defined as a code with $2^{nR}$ codewords $(x^n(m))_{m=1}^{2^{nR}}$, such that $\sum_{t=1}^nx_t^2(m)\leq np$, for every $m=1,\ldots,2^{nR}$ \cite{cover,ElGamalK_book}. However, instead of an average power constraint on each codeword,  we can put an average power constraint on each transmitted symbol and require that $\E[x_t^2(M)]\leq P$, for $t=1,\ldots,n$ \cite{ShamaiB:95}.  Note that for a given code, the randomness in $\E[x_t^2(M)]\leq P$ is only due to the message $M$.  This alternative definition does not affect the capacity of the channel from $C=0.5\log(1+P/N)$. In this paper, we consider the latter definition because of some technical issues in the proof of the main result.

The equivalence of the  capacities corresponding to the two definitions can be shown as follows. The converse of the  capacity theorem stated in \cite{cover} applies to the case symbol-by-symbol power constraint as well.   For the achievability, consider the code construction presented in \cite{cover} with the same encoding and decoding strategy. For each $t=1,\ldots,n$, $\P(\E[x_t^2(M)]>P)=\P(2^{-nR}\sum_{m=1}^{2^{nR}}x_t^2(m)>P)\leq 2^{-2^{nR}\d(\e)}$, where $\d(\e)\to 0$ as $\e\to 0$. Hence, by the union bound $\P(\E[x_t^2(M)]>P,\;{\rm for}\;{\rm some}\;t) \leq n2^{-2^{nR}\d(\e)}$. This shows that   there exist  a  sequence of codes that both satisfy the power constraint on each coordinate and also have arbitrary  small probability of error. (The analysis of the probability of error presented in \cite{cover} applies here too.)
\end{remark}

%In this section, for a technical reason which is made clear in the proof, we restrict the joint source-channel codes to the set of admissible codes which are defined as follows. 

%\begin{definition}[Admissible codes]
%Consider a joint source-channel code of  source blocklength $L$ and channel blocklength $n$ for network $\Nc$. For each $t=1,2,\ldots,n$, define
%\[
%\d_t^{(a\to b)}(x_{t},y_{t})\triangleq\E[d_L(U^{(a),L},\Uh^{(a\to b),L})|(X_t,Y_t)=(x_t,y_t)].
%\]
%The code is called \emph{admissible}, if for every fixed $x_{t}\in\mathds{R}$,  and $(a,b)\in\Vc^2$, $\d_t^{(a\to b)}(x_{t},y_{t})$  is a non-decreasing function of $|y_{t}-x_{t}|$. 
%\end{definition}
%In other words, a code is admissible if its performance does not improve, \ie stays the same or deteriorates, as the absolute value of the Gaussian noise increases. Let $\Dc_a(\kappa,\Nc)$ denote the set of  distortion matrices that are achievable over network $\Nc$ using  admissible codes. Here $\Dc_a(\kappa,\Nc)\subseteq\Dc(\kappa,\Nc)$; whether the two regions are identical remains an open problem for future work. 
%

\begin{theorem}\label{thm:awgn2}
For a wireline network consisting of  discrete or AWGN point-to-point channels, 
\[
 \Dc(\kappa,\Nc_b)= \Dc(\kappa,\Nc).
\]
\end{theorem}
\begin{proof}[Proof of Theorem \ref{thm:awgn2}]
The second inclusion is immediate since the first part of the proof of Theorem \ref{thm:main} applies equally well for continuous channels case. To prove the first inclusion, we employ the discretization method used in \cite{McEliece:77}. Let network $\Nc^{(\jv,\kv)}$, with $\jv=(j_1,j_2,\ldots,j_n)$ and $\kv=(k_1,k_2,\ldots,k_n)$,  denote the network derived from network $\Nc$ by replacing the AWGN channel from $a$ to $b$ by the structure shown in  Fig.~\ref{fig:awgn_g}. The given channel relies on a pair of quantizers  $Q[j]$ and $Q[k]$ parametrized by indices $j$ and $k$. We allow the quantizer parameters to vary with $t$, setting $j=j_t$ and $k=k_t$ for each time $t\in\{1,2,\ldots,n\}$.  The quantizer $Q[i]$ is defined as follows. For $i\in\{1,2,\ldots\}$, let $\Delta=1/\sqrt{i}$, and define the quantizer $Q[i]$ with quantization levels $\Lc_i=\{-i\Delta,-(i-1)\Delta,\ldots,-\Delta,0,\Delta,\ldots,(i-1)\Delta,i\Delta\}$. For any $x\in\mathds{R}$, $Q[i]$ maps $x$ to $[x]_i$, which is the closest number to $x$ in $\Lc_i$ such that $|[x]_i|\leq x$. Note that by this definition,  $\E[[X]_i^2]\leq \E[X^2]$ for any random variable $X$.

Lemma \ref{lemma:cont} in Appendix B shows that as $j$ and $k$ increase, the set of achievable distortions on  $\Nc^{(\jv,\kv)}$ approaches the set of achievable distortions on the original network. More precisely,
\begin{align}
\Dc_a(\kappa,\Nc)\subseteq \overline{\limsup_{\jv,\kv}\Dc(\kappa,\Nc^{(\jv,\kv)})},\label{eq:lemma1}
\end{align}
where 
\begin{align*}
\limsup_{\jv,\kv}\Ac_{\jv,\kv}\triangleq \bigcap_{\jv_0,\kv_0} \bigcup_{\substack{\jv\geq \jv_0\\ \kv\geq \kv_0}}\Ac_{\jv,\kv},
\end{align*}
and $\overline{\Ac}$ denotes the closure of the set $\Ac$.

We next show that
\begin{align}
\Dc(\kappa,\Nc^{(\jv,\kv)}) \subseteq \Dc(\kappa,\Nc_b).\label{eq:1}
\end{align}

This is sufficient to obtain the desired result since \eqref{eq:lemma1} and \eqref{eq:1} together imply $\Dc_a(\kappa,\Nc)\subseteq\Dc(\kappa,\Nc_b)$ by the closure in the definition of $ \Dc(\kappa,\Nc_b)$.

To prove that $\Dc(\kappa,\Nc^{(\jv,\kv)}) \subseteq \Dc(\kappa,\Nc_b)$, note that, at each time $t$,  the structure shown in Fig.~\ref{fig:awgn_g} behaves like a DMC with input $[X]_{j_t}$, power constraint  $\E[([X]_{j_t})^2]\leq P$ and output $[Y_{j_t}]_{k_t}$. Hence, by straightforward extension of the proof of Theorem \ref{thm:main}, 
\[
\Dc(\kappa,\Nc^{(\jv,\kv)})\subseteq\Dc(\kappa,\Nc_b^{(\jv,\kv)}),
\]
where $\Nc_b^{(\jv,\kv)}$ is identical to $\Nc^{(\jv,\kv)}$ except that the  channel from $a$ to $b$ is replaced by a bit pipe of capacity $C_{\jv,\kv}$ equal to the maximum capacity of the $n$ DMCs. Here
\begin{align*}
C_{\jv,\kv}\triangleq\max_{1\leq t\leq n}\max_{\substack{[X]_{j_t}\sim p_X\\ p_X:\E[X^2]\leq P}}I([X]_{j_t};[Y_{j_t}]_{k_t}).
\end{align*} 

By the data processing inequality \cite{cover},
\begin{align*}
I([X]_{j_t};[Y_{j_t}]_{k_t})&\leq I([X]_{j_t};Y_{j_t})\nonumber\\
&=h(Y_{j_t})-h(Z).
\end{align*}
On  the other hand, by the construction of the quantizers, 
\begin{align*}
\E[Y_{j_t}^2]&=\E[[X]_{j_t}^2]+N\nonumber\\
&\leq \E[X^2]+N.
\end{align*}
Hence, 
\begin{align*}
h(Y_{j_t})\leq 0.5\log(2\pi e (P+N)),
\end{align*}
and as a result
\begin{align*}
I([X]_{j_t};[Y_{j_t}]_{k_t})&\leq C.
\end{align*}
Therefore,  $\Dc(\kappa,\Nc^{(\jv,\kv)}) \subseteq \Dc_b$.

\begin{figure}[h]
\begin{center}
\psfrag{Q1}[l]{$Q[j]$}
\psfrag{Q2}[l]{$Q[k]$}
\psfrag{Z}[l]{$Z$}
\psfrag{X}[l]{$X$}
\psfrag{Y}[l]{$Y_j$}
\psfrag{Xq}[l]{$[X]_j$}
\psfrag{Yq}[l]{$[Y_j]_k$}
\includegraphics[width=7.5cm]{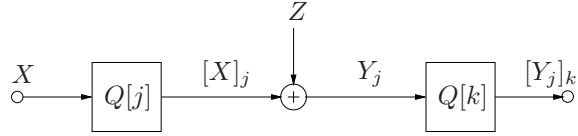}\caption{Quantizing the input and output alphabets of an AWGN}\label{fig:awgn_g}
\end{center}
\end{figure}

\end{proof}

%%******************************************************************************
%%******************************************************************************
%%
%\section{Sources with memory}\label{sec:memory}
%\input{memory}
%
%******************************************************************************
%******************************************************************************

\section{Conclusions}\label{sec:conclusion}

In this paper we proved the separation of source-network and channel coding in  general wireline networks of independent  discrete  point-to-point channels with dependent sources and arbitrary lossy or lossless reconstruction demands. We also proved that the  result continues to hold when one or more channels is an AWGN channel.
%------------------Appendix II -----------------
\renewcommand{\theequation}{A-\arabic{equation}}
% redefine the command that creates the equation no.
\setcounter{equation}{0}  % reset counter
\section*{Appendix A: Proof of part ii of Theorem \ref{thm:stack}}

 Let $D\in {\rm int} (\Dc_s(\kappa,\underline{\Nc}))$. Then for any $\e>0$, there exist integers $N$, $n$, and $L$ such that $L/n \geq \kappa-\e$ and there exists a  blocklength-$n$ coding scheme for $L$ source symbols on $N$-fold  stacked network $\underline{\Nc}$ that achieves 
    \[
    \E \left[d^{(a\to b)}_{NL}(U^{(a),NL},\hat{U}^{(a\to b),NL})\right] \leq D(a,b) + \e,
    \]
 for all $a,b\in\Vc$. The same coding scheme can be used in a single-layer network as follows. Consider a single layer network where each node $a$ observes a length-$NL$ block of source symbols $U^{(a),NL}$ and describes the block in the next $Nn$ time steps. Given source blocklength $L'=NL$ and channel block length $n'=Nn$, the code has rate $\kappa=L'/N'=L/N$. At each time $t\in\{1,\ldots,N\}$, each node $a\in\Vc$ sends, over its outgoing edges, what it would have  sent at time 1 in layer $t$ of $\underline{\Nc}$, \ie $\Xv^{(a)}_1(t)$, and collects, over its incoming edges, what it would have collected in layer $t$ of $\underline{\Nc}$, \ie  $\Yv^{(a)}_1(t)$. At times $t\in\{N+1,\ldots,2N\}$, each node $a$ sends $\Xb_2^{(a)}(t-N)$ and collects $\Yb_2^{(a)}(t-N)$.  Here calculating $\Xb_2^{(a)}$ is possible due to the prior collection of $\Yb_2^{(a)}$. 
 %(Note that in the first $N$ time steps, node $a$'s output is only a function of its own source, not the channels' outputs. It only collects the channel outputs in order to use them during the next $N$ time steps.) 
The same strategy is used in the next $n-2$ time intervals,  in interval $t$ transmitting $\Xv_t^{(a)}$  for $t\in\{3,\ldots,n\}$ and collecting $\Yv_t^{(a)}$ for uses in future time intervals. Using this strategy, at the end of $nN$ channel uses, each node's observation has exactly the same distribution as the collection of observations of its $N$ copies in the stacked networks. Therefore, applying the  decoding rules  results in the same distortion. Hence, $D\in\Dc(\kappa,\Nc)$.

%------------------Appendix II -----------------
\renewcommand{\theequation}{B-\arabic{equation}}
% redefine the command that creates the equation no.
\setcounter{equation}{0}  % reset counter

\section*{Appendix B: Lemma \ref{lemma:cont}}
\begin{lemma}\label{lemma:cont}
For any $\kappa>0$,
\begin{align}
\Dc(\kappa,\Nc)\subseteq \overline{\limsup_{\jv,\kv}\Dc(\kappa,\Nc^{(\jv,\kv)})},
\end{align}
where $\overline{\Ac}$ denotes the closure of set $\Ac$.
\end{lemma}
\begin{proof}
%Any performance achievable on $\Nc^{(\jv,\kv)}$ can be achieved on $\Nc$ by adding the proper quantizers to the input and output of the channel from $a$ to $b$ in $\Nc$. Hence, \[\overline{\limsup_{\jv,\kv}\Dc(\kappa,\Nc^{(\jv,\kv)})}\subseteq\Dc(\kappa,\Nc).\] 
%%
%The nontrivial step is  to prove that  $\Dc(\kappa,\Nc)\subseteq\overline{\limsup_{\jv,\kv}\Dc(\kappa,\Nc^{(\jv,\kv)})}$. 
Let $D\in \Dc(\kappa,\Nc)$. For any $\e>0$, and for $L$ sufficiently large, there exist a joint source-channel code at rate $\kappa$ with source blocklength $L$ such that 
\begin{align}
\E[d_L(U^{(a),L},\hat{U}^{(a\to b),L})]\leq D(a,b)+\e,
\end{align}
 holds for each $(a,b)\in\Vc^2$. Let $U^L=U^{(a),L}$ and $\Uh^L=\Uh^{(a\to b),L}$ for some  fixed $(a,b)\in\Vc^2$. 
 
Conditioning the expected average distortion between $U^L$ and $\Uh^L$ on the input and output values of the AWGN channel at time $t=1$, it follows that 
\begin{align}
D(a,b)+\e&\geq \E[d_L(U^L,\Uh^L)]\nonumber\\
&=\sum_{(x_1,y_1)}p(x_1,y_1)\E[d_L(U^L,\Uh^L)|(X_1,Y_1)=(x_1,y_1)]\nonumber\\
&=\E[\d^{(1)}(X_1,Y_1)]
\end{align}
where $\d^{(1)}(x_1,y_1)\triangleq\E[d_L(U^L,\Uh^L)|(X_1,Y_1)=(x_1,y_1)]$.\\
Now assume that  the same code is applied to network $\Nc^{(j_1,k_1)}$, which is identical to $\Nc$ except that at  time $t=1$,   the AWGN channel is replaced by the structure shown in Fig.~\ref{fig:awgn_g} with parameters $j=j_1$ and $k=k_1$. The expected average distortion between $U^L$ and $\Uh^L$ in the modified network,  $D^{(j_1,k_1)}(a,b)$, can be written as
\begin{align}
D^{(j_1,k_1)}(a,b)=\E[\d^{(1)}(X_1,\Yt_1)], 
\end{align}
where $\Yt_1\triangleq[[X]_{j_1}+Z_1]_{k_1}$. Note that, conditioned on the input and output values of the AWGN channel at time $t=1$, the two networks have identical performance. \\ Further, $\Yt_1$  converges pointwise to $Y_1$ almost everywhere as $j$ and $k_1$ grow without bound, \ie
\begin{align}
\lim_{k_1\to\infty}\lim_{j_1\to\infty}\Yt_1=Y_1,
\end{align}
almost everywhere, where $Y_1=X+Z$. 

While function $\d^{(1)}(x_1,y_1)$ might not be continuous everywhere, by the Lusin's Theorem \cite{kechris1995classical}, since it is  measurable, for any $\e_1>0$, there exists a subset $\Ac_1\subset\mathds{R}^2$, such that $\P((X_1,Y_1)\in\Ac_1)<\e_1$ and $\d^{(1)}$ is continuous on $\Ac_1$.  By the law of iterated expectations,
\begin{align}
D^{(j_1,k_1)}(a,b)&=\E[\d^{(1)}(X_1,\Yt_1)|(X_1,Y_1)\in\Ac_1]\P((X_1,Y_1)\in\Ac_1)\nonumber\\
&\;\;\;\;+\E[\d^{(1)}(X_1,\Yt_1)|(X_1,Y_1)\notin\Ac_1]\P((X_1,Y_1)\notin\Ac_1)\nonumber\\
&\leq \E[\d^{(1)}(X_1,\Yt_1)|(X_1,Y_1)\in\Ac_1]+d_{\max}\e_1,\label{eq:step1}
\end{align}
and, similarly,  
\begin{align}
D^{(j_1,k_1)}(a,b)
&\geq \E[\d^{(1)}(X_1,\Yt_1)|(X_1,Y_1)\in\Ac_1]\P((X_1,Y_1)\in\Ac_1)\nonumber\\
&\geq  \E[\d^{(1)}(X_1,\Yt_1)|(X_1,Y_1)\in\Ac_1](1-\e_1).\label{eq:step2}
\end{align}

Since $\d(x_1,y_1)$ is continuous on $\Ac_1$ and is bounded,  by the bounded convergence theorem, it follows from  \eqref{eq:step1} and  \eqref{eq:step2} that
\begin{align}
\lim_{k_1\to\infty}\lim_{j_1\to\infty}D^{(j_1,k_1)}(a,b)
&\leq\lim_{k_1\to\infty}\lim_{j_1\to\infty}  \E[\d^{(1)}(X_1,\Yt_1)|(X_1,Y_1)\in\Ac_1]+d_{\max}\e_1\nonumber \\
&= \E[\lim_{k_1\to\infty}\lim_{j_1\to\infty} \d^{(1)}(X_1,\Yt_1)|(X_1,Y_1)\in\Ac_1]+d_{\max}\e_1\nonumber \\
&\leq \E[ \d^{(1)}(X_1,Y_1)|(X_1,Y_1)\in\Ac_1]+d_{\max}\e_1,\label{eq:step3}
\end{align}
and
\begin{align}
\lim_{k_1\to\infty}\lim_{j_1\to\infty}D^{(j_1,k_1)}(a,b)&\geq \E[\d^{(1)}(X_1,Y_1)|(X_1,Y_1)\in\Ac_1](1-\e_1). \label{eq:step4}
\end{align}

On the other hand, 
\begin{align}
\E[ \d^{(1)}(X_1,Y_1)] &\geq \E[ \d^{(1)}(X_1,Y_1)|(X_1,Y_1)\in\Ac_1]\P((X_1,Y_1)\in\Ac_1)\nonumber\\
&\geq \E[ \d^{(1)}(X_1,Y_1)|(X_1,Y_1)\in\Ac_1](1-\e_1),\label{eq:step5}
\end{align}
and
\begin{align}
\E[\d^{(1)}(X_1,Y_1)] &\leq \E[ \d^{(1)}(X_1,Y_1)|(X_1,Y_1)\in\Ac_1]+d_{\max}\e_1.\label{eq:step6}
\end{align}
Since $\E[\d^{(1)}(X_1,Y_1)]=\E[d_L(U^L,\Uh^L)]$, combining \eqref{eq:step3} and \eqref{eq:step4} with \eqref{eq:step5} and \eqref{eq:step6} yields
\begin{align}
\E[d_L(U^L,\Uh^L)] -d_{\max}\e_1 \leq \lim_{k_1\to\infty}\lim_{j_1\to\infty}D^{(j_1,k_1)}(a,b)\leq {\E[d_L(U^L,\Uh^L)]\over 1-\e_1}+d_{\max}\e_1.\label{eq:sandwich}
\end{align}
Since $\e_1$ can be made arbitrary small, from \eqref{eq:sandwich}, we have 
\begin{align}
 \lim_{k_1\to\infty}\lim_{j_1\to\infty}D^{(j_1,k_1)}(a,b)=\E[d_L(U^L,\Uh^L)].\label{eq:t=1}
\end{align}

The prior analysis captures the expected distortion when the continuous channel is replaced by a finite alphabet channel only at time 1. 
To finish the proof we use induction. Assume that for times $1,2,\ldots,t-1$, the continuous channel can be replaced a finite alphabet channel without, asymptotically, changing the expected average distortion, \ie
\begin{align}
&\lim_{k_{1}\to\infty}\lim_{j_{1}\to\infty}\ldots\lim_{k_{t-1}\to\infty}\lim_{j_{t-1}\to\infty}D^{(j^{t-1},k^{t-1})}(a,b)\nonumber\\
&=\E[\d^{(t-1)}(X^{t-1},Y^{t-1})],\nonumber\\
&=\E[d_L(U^L,\Uh^L)].\label{eq:ind-assumption}
\end{align}
where $D^{(j^{t-1},k^{t-1})}(a,b)$ denotes the expected average distortion between $U^L$ and $\Uh^L$ in the modified network when the parameters of the channel input and output  quantizers, at times $\tau\in\{1,\ldots,t-1\}$, are $(j^{t-1},k^{t-1})$, and let
\[
\d^{(t-1)}(x^{t-1},y^{t-1})\triangleq\E[d_L(U^L,\Uh^L)|(X^{t-1},Y^{t-1})=(x^{t-1},y^{t-1})].
\]
Now we need to show that if we add the quantizers at time $t$ as well, the performance does not change. 

In the original network
\begin{align}
\E[d_L(U^L,\Uh^L)]&=\E[\d^{(t)}(X^{t},Y^{t})]\nonumber\\
&=\E\left[\E[\d^{(t)}(X^{t},Y^{t})]|(X^{t},Y^t)]\right], 
\end{align}
and in the modified network,
\begin{align}
D^{(j^t,k^t)}(a,b)&=\E[\d^{(t)}(\Xt^t,\Yt^t)], \nonumber\\
&=\E\left[\E\left[\left.\d^{(t)}(\Xt^t,\Yt^t)\right|(\Xt^{t-1},\Yt^{t-1})\right]\right],\label{eq:D2}
\end{align}
where for $t'\in\{1,\ldots,t\}$, $\Xt_{t'}$ is the channel input at time $t'$ when the given code is applied and the Gaussian channel replaced by its quantized approximation, and  $\Yt_{t'}=[[\Xt_{t'}]_{j_{t'}}+Z_{t'}]_{k_{t'}}$. Note that $\Xt_1=X_1$. 

While $X_t$ and $\Xt_t$ might have different distributions due to the quantizations at times $t'=1,\ldots,t-1$, their conditional distributions given the inputs and outputs of the channel up to time $t-1$ are identical in both networks, \ie
\begin{align}
&\P\left(\left.\Xt_t<x_t\right|(\Xt^{t-1},\Yt^{t-1})=(x^{t-1},y^{t-1})\right)\nonumber\\
&=\P\left(\left.X_t<x_t\right|(X^{t-1},Y^{t-1})=(x^{t-1},y^{t-1})\right).\label{eq:F}
\end{align}
Let 
\begin{align}
\gamma&(x^{t-1},y^{t-1})\triangleq\E[\d^{(t)}(X^t,Y^t)]|(X^{t-1},Y^{t-1})=(x^{t-1},y^{t-1})]\nonumber\\
&\E\left[\int\d^{(t)}(x^t,y^{t-1},Y_t)dF(x_t|(x^{t-1},y^{t-1}))\right]
\end{align}
and 
\begin{align}
&\tilde{\gamma}^{(j_t,k_t)}(x^{t-1},y^{t-1})\triangleq\E[\d^{(t)}(\Xt^t,\Yt^t)]|(\Xt^{t-1},\Yt^{t-1})=(x^{t-1},y^{t-1})]\nonumber\\
&=\E\left[\int\d^{(t)}(x^t,y^{t-1},\Yt_t)dF(x_t|(x^{t-1},y^{t-1}))\right],
\end{align}
where in the last line we are using \eqref{eq:F}.

Using the same argument as the one used to prove \eqref{eq:t=1}, it follows that
\begin{align}
&\lim_{k_t\to\infty} \lim_{j_t\to\infty}\tilde{\gamma}^{(j_t,k_t)}(x^{t-1},y^{t-1})=\gamma(x^{t-1},y^{t-1}).\label{eq:kj-t}
\end{align}
Hence,
\begin{align}
&\lim_{k_1\to\infty} \lim_{j_1\to\infty}\ldots\lim_{k_t\to\infty} \lim_{j_t\to\infty}D^{(j^t,k^t)}(a,b)\nonumber\\
&=\lim_{k_1\to\infty} \lim_{j_1\to\infty}\ldots\lim_{k_t\to\infty} \lim_{j_t\to\infty}\E[\tilde{\gamma}^{(j_t,k_t)}(\Xt^{t-1},\Yt^{t-1})]\nonumber\\
&=\lim_{k_1\to\infty} \lim_{j_1\to\infty}\ldots\lim_{k_{t-1}\to\infty} \lim_{j_{t-1}\to\infty}\E\left[\lim_{k_t\to\infty} \lim_{j_t\to\infty}\tilde{\gamma}^{(j_t,k_t)}(\Xt^{t-1},\Yt^{t-1})\right]\nonumber\\
&\stackrel{(a)}{=}\lim_{k_1\to\infty} \lim_{j_1\to\infty}\ldots\lim_{k_{t-1}\to\infty} \lim_{j_{t-1}\to\infty}\E\left[{\gamma}(\Xt^{t-1},\Yt^{t-1})\right]\nonumber\\
%&=\lim_{k_1\to\infty} \lim_{j_1\to\infty}\E[\lim_{k_2\to\infty} \lim_{j_2\to\infty}\tilde{\gamma}^{(j_2,k_2)}(X_1,\Yt_1)]\nonumber\\
%&=\lim_{k_1\to\infty} \lim_{j_1\to\infty}\E[\gamma(X_1,\Yt_1)]\nonumber\\
%&=\E\left[\lim_{k_1\to\infty} \lim_{j_1\to\infty}\gamma(X_1,\Yt_1)\right]\nonumber\\
%&=\E[\gamma(X_1,Y_1)]\nonumber\\
&\stackrel{(b)}{=}\E[d_L(U^{(a),L},\hat{U}^{(a\to b),L})],
\end{align}
where $(a)$ follows from \eqref{eq:kj-t} plus the dominated convergence theorem, and (b) follows from our inductive hypothesis.

\end{proof}

\section*{Acknowledgments}
This work was supported in part by Caltech's Center for the Mathematics of Information (CMI) and DARPA ITMANET grant W911NF-07-1-0029.

\bibliographystyle{unsrt}
\bibliography{myrefs}

\end{document}